\documentclass[12pt]{article}
\usepackage[margin=1.95cm]{geometry}

\usepackage{color}
\usepackage[normalem]{ulem}
\usepackage{siunitx}
\usepackage{amsmath}
\usepackage{amsfonts}
\usepackage{amssymb}
\usepackage{amsthm}
\usepackage{graphicx}
\usepackage{times}
\usepackage{bm}
\usepackage{natbib}
\usepackage[plain,noend]{algorithm2e}
\usepackage{xcolor}
\usepackage{multirow}
\usepackage{lscape}
\usepackage{booktabs}
\usepackage{threeparttable}
\usepackage[bottom]{footmisc}

\usepackage{hyperref}
\usepackage{setspace}
\usepackage{textcomp}
\usepackage{enumitem}

\newcommand{\mR}{\mathbb{R}}

\newcommand{\E}{\mathbb{E}}
\newcommand{\PP}{{P}}

\newcommand{\OP}{\mathcal{O}_{P}}
\newcommand{\oP}{o_{P}}

\newcommand{\dell}{\dot{\ell}}
\newcommand{\ddell}{\ddot{\ell}}
\newcommand{\hbeta}{\widehat{\beta}}
\newcommand{\hTheta}{\widehat{\Theta}}
\newcommand{\hSigma}{\widehat{\Sigma} }

\newcommand{\Thetabeta}{\Theta_{\beta^0}}
\newcommand{\hmu}{\widehat{\mu}}
\newcommand{\heta}{\widehat{\eta}}

\newtheorem{theorem}{Theorem}
\newtheorem{corollary}[theorem]{Corollary}
\newtheorem{lemma}[theorem]{Lemma}

\begin{document}

\title{\bf Statistical Inference for  Cox Proportional Hazards Models with a Diverging Number of Covariates}

\author{Lu Xia\textsuperscript{a}, Bin Nan\textsuperscript{b}\footnotemark[1], and Yi Li\textsuperscript{c}\footnotemark[1]\\
			\small 
			\textsuperscript{a}Department of Biostatistics, University of Washington, Seattle, WA, {xialu@uw.edu}  \\
			\small
			\textsuperscript{b}Department of Statistics, University of Californina, Irvine, CA, {nanb@uci.edu} \\
			\small
			\textsuperscript{c}Department of Biostatistics, University of Michigan, Ann Arbor, MI, {yili@umich.edu}
		}

\date{}

\maketitle

\footnotetext[1]{To whom correspondence should be addressed.}

\begin{abstract}
For statistical inference on regression models with a diverging number of covariates, the existing literature typically makes  sparsity assumptions on the inverse of the Fisher information matrix. Such  assumptions, however, are often violated under  Cox proportion hazards models, leading to biased estimates with under-coverage confidence intervals. We propose a modified debiased lasso approach, which solves a series of quadratic programming problems to approximate the inverse information matrix without posing sparse matrix assumptions. We establish asymptotic results for the estimated regression coefficients when the  dimension of covariates diverges with the sample size. As demonstrated by extensive simulations, our proposed method provides consistent estimates and confidence intervals with nominal coverage probabilities. The utility of the method is further demonstrated by assessing the effects of genetic markers on patients' overall survival with the Boston Lung Cancer Survival Cohort, a large-scale epidemiology study investigating mechanisms underlying the lung cancer.\\[0.1cm]
	
\end{abstract}

%%%%%%%%%%%%%%%%%%%%%%%%%%%%%%%%%%%%%%%%%

\section{Introduction}

The Cox proportional hazards model \citep{cox1972regression}, a semiparametric model with an unspecified baseline hazard function,  has been widely used for the analysis of censored time-to-event data. With a fixed dimension of covariates,  \citet{cox1972regression} proposed the maximum partial likelihood estimation (MPLE) to infer the regression coefficients, and \citet{andersen1982cox} proved the asymptotic distributional results for  MPLE using the Martingale theory. 

Technological advances nowadays have made it possible to collect a large amount of information in biomedical studies. For example, the Boston Lung Cancer Survival Cohort (BLCSC), the motivating study for this work, has acquired abundant clinical, genetic, epigenetic and genomic data, which enable comprehensive investigations of molecular mechanisms underlying the lung cancer survival \citep{mckay2017large}. High-dimensionality of the collected covariates  has confronted the traditional parameter estimation and uncertainty quantification based on Cox models. In high-dimensional settings, where the number of covariates $p$ increases with the sample size $n$ or even greater than $n$, the conventional maximum partial likelihood estimation is usually ill-conditioned. Penalized estimators have emerged as a powerful tool for simultaneous variable selection and estimation \citep{tibshirani1997lasso,fan2002variable,gui2005penalized,antoniadis2010dantzig}.  Recently, \citet{huang2013oracle} and \citet{kong2014non} derived  the non-asymptotic oracle inequalities of the lasso estimator in the Cox model. However, none of these works dealt with statistical inference for
Cox models with high-dimensional covariates.
%where additional difficulties emerged since the negative log partial likelihood loss function is not a sum of independent and identically distributed terms nor Lipschitz. 

Existing literature on inference for high-dimensional models mainly concerns linear regression. \citet{zhang2014confidence}, \citet{van2014asymptotically} and \citet{javanmard2014confidence} developed inference procedures for linear models, based on debiasing the lasso estimator via low-dimensional projection or inverting the Karush--Kuhn--Tucker condition.  \citet{van2014asymptotically} extended the debiased lasso idea to generalized linear models, using the nodewise lasso regression. \citet{ning2017general} focused on hypothesis testing and devised decorrelated score, Wald and likelihood ratio tests for inference on a low-dimensional parameter in generalized linear models based on projection theory.  

There has been limited progress in inference for the Cox model with  high-dimensional covariates. \citet{fang2017testing} developed decorrelated tests for hypothesis testing of low-dimensional components under  high-dimensional Cox models, using ideas similar to \citet{ning2017general}. \citet{kong2018high} extended the debiased lasso approach in \citet{van2014asymptotically} to potentially misspecified Cox models, and used the nodewise lasso regression to estimate the inverse information matrix. \citet{yu2018confidence} proposed a debiased lasso approach, by estimating the inverse information matrix with a CLIME estimator adapted from \citep{cai2011constrained}. 
%which was originally designed for precision matrix estimation. 
%these  methods have certain limitations. Due to the properties of the nodewise lasso regression and CLIME, both \citet{kong2018high} and \citet{yu2018confidence} 
Most of these works restricted the number of non-zero elements of each row in the inverse information matrix to be small, i.e. $\ell_0$ sparsity. However, as found in \citet{xia2020revisit}, 
%it has been argued that such a sparsity assumption on the high-dimensional inverse information matrix does not hold in general settings of generalized linear models. This is because, in generalized linear models, the observed information matrix takes the form of $X^T W_{\beta^0} X$, where $X$ is the design matrix and $W_{\beta^0}$ is a diagonal matrix with 
%the response variances on the diagonal, 
%each diagonal element being a function of covariates, which distorts the interpretability and the validity of sparsity assumption on the inverse of its expectation in general settings. The same argument is also applicable to the Cox model for the same reason, i.e.,  %\citet{fang2017testing} imposed an $\ell_0$ sparsity assumption on $w^* = H_{\theta \theta}^{* -1} H_{\theta \alpha}^*$, where, in their notation, $H^*$ is the Fisher information matrix, $\alpha$ is the low-dimensional component of interest, $\theta$ is the nuisance parameter, and $w^*$ is approximated using a lasso-type estimator. Imposing these sparsity conditions plays an important role in studying the theoretical properties of the aforementioned estimators, and it is hard to estimate such a high-dimensional matrix well without them under the ``large $p$, small $n$" scenario. However, these
the sparse inverse information matrix assumption has no practical interpretation  beyond linear regression models, often fails to  hold in the Cox model, and these methods cannot perform satisfactorily in high-dimensional Cox model settings. For example,  as evidenced by our extensive simulations,  these methods cannot correct biases of lasso estimators or  construct confidence intervals with desired coverage probabilities, even when the number of regression coefficients is moderate relative to the sample size.

Our work is pertaining to  the  ``large $n$, diverging $p$" framework where $p < n$ and $p$ is allowed to increase with $n$ to infinity, which reflects the setting of the  motivating BLCSC {with $n=561$ and $p=231$}. Under this framework,
we draw inference based on  Cox models without imposing sparsity to the inverse information matrix.  
%Our primary focus is on improving the bias correction and delivering more reliable confidence intervals and inference results. \citet{xia2020revisit}  suggested directly inverting the information matrix after variable selection in generalized linear models where the covariate dimension $p$ grows with the sample size $n$ and satisfies $p^2\log(p)/n \rightarrow 0$ as $n \rightarrow \infty$. 
% under certain mild conditions. 
%Our numerical exploration, however, indicates that directly applying this approach to the Cox model does not adequately correct biases, hence some finer tuning for the inverse information matrix estimation is warranted.  
%Inspired by \citet{javanmard2014confidence}, 
Specifically, we propose a debiased lasso approach via solving a series of quadratic programming problems to estimate the inverse information matrix. We use quadratic programming as a means of balancing the bias-variance trade-off and avoiding the unrealistic $\ell_0$ sparsity assumption for the large inverse information matrix in the Cox model. Our work adds to the literature in the following aspects.
%with both theoretical and practical contributions. 
First, unlike \citet{javanmard2014confidence}, our work entails careful treatment of the sum of non independently nor identically distributed terms in the empirical loss function, and we consider random designs instead of deterministic designs. Second, we find that the tuning parameter selection for the inverse information matrix estimation is crucial for bias correction. {For example, a related work  \citep{yu2018confidence}
proposed to select tuning parameters by minimizing {the cross-validated difference between the product of the information matrix with its estimated inverse and the identity matrix,} but was found to perform poorly. In contrast, we propose a cross-validation procedure to tune parameters by  hard thresholding  debiased estimates when
 %cross-validation procedure using debiased estimators after hard thresholding is implemented to select the important tuning parameter in each of 
 solving the quadratic programming problems,  which yields} satisfactory numerical performance.
%On the practical side, this also distinguishes our work from \citet{yu2018confidence}. \bincm{What do you mean here? Anything specific in Yu et al?} \lucm{This last sentence may look like diminishing the contribution of our work, and I can remove it. But what I was trying to say is that \citet{yu2018confidence} implemented a cross-validation criterion based on minimizing $trace(\widehat{\Theta} \widehat{\Sigma} - I_p)^2$, which performed poorly in practice for bias correction.}

The article is organized as follows. Section \ref{sec:method} introduces the proposed debiased lasso approach, where the inverse information matrix is estimated via quadratic programming  with a novel cross-validation procedure for selecting the tuning parameter. 
Section \ref{sec:theory}  lays the theoretical foundation for reliable inference on linear combinations of the Cox regression parameters using debiased lasso estimators. We examine the finite sample performance of our proposed method with simulation studies in Section \ref{sec:simulation}, apply it to analyze the  BLCSC data in Section \ref{sec:app}, and conclude the paper with a few remarks in Section \ref{sec:conclusion}. We state several useful technical lemmas and provide proofs of the main theorems in the Appendix, and defer proofs of all the lemmas to the online supplementary materials.

%===============================================
\section{Method}
\label{sec:method}

\subsection{Background and set-up}

We introduce  notation that will be used throughout this article. For a vector $x = (x_1, \ldots, x_r)^T \in \mR^r$, $x^{\otimes 0} = 1$, $x^{\otimes 1} = x$ and $x^{\otimes 2} = x x^T$. The $\ell_q$-norm for  $x$ is $\| x \|_q = (\sum_{j=1}^r |x_j|^q)^{1/q}$, $q \ge 1$, and the $\ell_0$-norm  is $\| x \|_0 = \sum_{j=1}^r I( x_j \ne 0)$.
%\bincm{What about $\ell_0$ norm? How do you make it consistent with $\ell_p$ norm?} \lucm{$\ell_p$ norm ($p \ge 1$) is a real norm satisfying the three conditions in the definition of a norm, but $\ell_0$ ``norm" is not a proper norm since $\| \alpha x \|_{0} = \| x\|_0$ for $\alpha \ne 0$. And if there are more than one $(>1)$ non-zero elements in $x$, $\lim_{q \to 0^+} \| x \|_q = \infty$, and only $\lim_{q \to 0^+} \sum_{j=1}^r | x_j |^q = \| x \|_0$. That's why I don't usually unify the two definitions.} 
For a matrix $A = (a_{ij}) \in \mR^{m \times r}$, the induced matrix norm is defined as $\| A \|_{q_1,q_2} = \sup_{x \in \mR^r, x \ne 0} \| Ax \|_{q_2} / \| x \|_{q_1}$, $q_1, q_2 \ge 1$. In particular, $\| A \|_{1,1} = \max_{1 \le j \le r} \sum_{i=1}^m |a_{ij}|$, $\| A \|_{2,2} = \sigma_{\mathrm{max}}(A)$, the largest singular value of $A$, and $\| A \|_{\infty,\infty} = \max_{1 \le i \le m} \sum_{j=1}^r |a_{ij}|$. The element-wise max norm is denoted as $\| A \|_{\infty} = \max_{i, j} |a_{ij}|$. {For two positive sequences $\{d_n\}$ and $\{g_n\}$, we define $d_n \asymp g_n$ if there are two bounded positive constants $C$ and $C^{\prime}$ such that $C \le d_n / g_n \le C^{\prime}$.} 

A Cox model stipulates  that the  hazard function for the underlying failure time $T$, conditional on a $p$-dimensional vector of covariates $X = (X^{(1)}, \ldots, X^{(p)}) \in \mR^p$, is $h(t | X) = h_0(t) \exp\{ X^T \beta^0 \}$,
where $h_0(t)$ is an unknown baseline hazard function and $\beta^0 = (\beta^0_1, \ldots, \beta^0_p)^T \in \mR^p$ is an unknown vector of regression coefficients.  With $T$ subject to right censoring, the observed survival time is  $Y = \min(T, C)$, where the censoring time $C$ is assumed to be independent of $T$ given  $X$. Let $\delta = 1(T \le C)$ denote the event indicator. Based on $n$ independent and identically distributed observations $\{ Y_i, X_i, \delta_i\}_{i=1}^n$, the goal of the paper is to estimate and draw inference on the regression coefficients $\beta^0$, when $p<n$ but $p \rightarrow \infty$ as $n \rightarrow \infty$.

\subsection{Debiasing the lasso estimator}

{When $p$ is fixed, a natural approach for inferring $\beta^0$ is through  maximum
 partial likelihood estimation (MPLE), which maximizes the log partial likelihood function 
 %\ylcm{ $\ell_n(\beta)$ no long a negative likelihood, please check the rest of the theory.} 
 %\lucm{To avoid making mistakes in too many places to be changed, I added a negative sign to $\ell_n(\beta)$ below. - yl: not a diligent response. I can only
 %do the best i can to help } 
\begin{equation} \label{eq:negloglik}
%\lu{-} \ell_n(\beta) = 
\displaystyle \frac{1}{n} \sum_{i=1}^n \left[    X_i^T \beta - \log \left\{ \displaystyle \frac{1}{n} \sum_{j=1}^n 1(Y_j \ge Y_i) \exp(X_j^T \beta)  \right\}  \right] \delta_i.
\end{equation}
%\lu{where $\ell_n(\beta)$ denotes the negative log partial likelihood.}
However, with a diverging $p$ of our interest, MPLE may suffer from numerical instability and yield unreliable inference; see Section \ref{sec:simulation}.
%\ylcm{cite sth} 
%\lucm{Cannot find papers with similar settings to ours that had numerical results on the instability of MPLE. - yl: you should have looked at other settings, e.g. in glms. anyway, too late. }
%We, therefore, propose a  bias correction method by borrowing ideas from debiased lasso. We remove biases from lasso estimates, instead of directly from MPLEs, is because lasso estimates tend to be numerically more stable than MPLEs \ylcm{cite sth}, though asymptotically both approaches retain the same convergence rate. To proceed, we briefly review debiased lasso, and then present our method.
}

A more commonly used approach, when  $p$ diverges to $\infty$ as $n \rightarrow \infty$, is a lasso estimator,  defined to be the minimizer of the following penalized negative log partial likelihood:
\begin{equation*}
\hbeta = \mathrm{argmin}_{\beta \in \mathbb{R}^p} \left\{ \ell_n(\beta) + \lambda_n \| \beta \|_1 \right\},
\end{equation*}
where  $\ell_n(\beta)$ is the negative  log partial likelihood function, i.e.
 the negative of (\ref{eq:negloglik}), and $\lambda_n > 0$ is  a tuning parameter to be decided.
%where the negative log partial likelihood function is
%\begin{equation} \label{eq:negloglik}
%\ell_n(\beta) = - \displaystyle \frac{1}{n} \sum_{i=1}^n \left[    X_i^T \beta - \log \left\{ \displaystyle \frac{1}{n} \sum_{j=1}^n 1(Y_j \ge Y_i) \exp(X_j^T \beta)  \right\}  \right] \delta_i.
%\end{equation}
%Then the lasso estimator $\hbeta$ is defined as the minimizer of the following penalized negative log partial likelihood:
%\begin{equation*}
%\hbeta \in \mathrm{argmin}_{\beta \in \mathbb{R}^p} \left\{ \ell_n(\beta) + \lambda_n \| \beta \|_1 \right\},
%\end{equation*}
%for some tuning parameter $\lambda_n > 0$. %\bincm{NOTE TO LU: should you use $\lambda_n$ to be consistent with $\gamma_n$? Or maybe consider using a different Greek letter because $\lambda_n$ is used to denote hazard function and eigenvalues. Try to use different letters for those different quantities.} \lucm{Changed $\lambda$ to $\lambda_n$, basedline hazard and cumulative hazard functions to $h_0(t)$ and $H_0(t)$, max and min eigenvalues to $\zeta_{max}$ and $\zeta_{min}$.}
The first and second order derivatives of $\ell_n(\beta)$  with respect to $\beta$, that is, the score function and the information matrix,  are respectively  denoted by
\begin{equation*}
\dell_n(\beta)  =  - \displaystyle \frac{1}{n} \sum_{i=1}^n \left\{ X_i - \displaystyle \frac{\hmu_1(Y_i; \beta)}{\hmu_0(Y_i; \beta)}  \right\} \delta_i, \quad
\ddell_n(\beta) =  \displaystyle \frac{1}{n} \sum_{i=1}^n \left\{  \displaystyle \frac{\hmu_2(Y_i; \beta)}{\hmu_0(Y_i; \beta)} - \left[ \displaystyle \frac{\hmu_1(Y_i; \beta)}{\hmu_0(Y_i; \beta)} \right]^{\otimes 2} \right\} \delta_i ,
\end{equation*}
where $\hmu_r (t; \beta) = n^{-1} \sum_{j=1}^n 1(Y_j \ge t) X_j^{\otimes r} \exp\{X_j^T \beta\}, ~ r = 0, 1, 2$. We also define the weighted average covariate vector 
$
\heta_n(t; \beta) = \hmu_1(t; \beta) / \hmu_0(t; \beta) = \sum_{j=1}^n 1 (Y_j \ge t) \exp\{ X_j^T \beta \} X_j \big{/}  \sum_{j=1}^n 1 (Y_j \ge t) \exp\{ X_j^T \beta \}.
$

The lasso estimates tend to be more stable  because of  the penalization.
However, as the lasso estimator $\hbeta$ incurs biases  \citep{javanmard2014confidence}, we consider a debiased lasso approach to remove its bias and draw inference. Analogous to \citet{van2014asymptotically} for generalized linear models, we define a debiased lasso estimator for $\beta^0$ as
\begin{equation} \label{eq:dslasso}
\widehat{b} = (\widehat{b}_1, \ldots, \widehat{b}_p)^T = \hbeta - \widehat{\Theta} \dot{\ell}_n(\hbeta),
\end{equation}
with $- \widehat{\Theta} \dot{\ell}_n(\hbeta)$ serving as the bias correction term, where $\widehat{\Theta}$ is an estimate of the inverse information matrix.  
{A reliable estimator, $\hTheta$, is important to ensure the validity of the method. However,  existing methods, most of which  rely on $\ell_0$ sparsity assumptions on the true inverse information matrix and use  nodewise lasso or CLIME to  estimate a sparse $\hTheta$,  are  found to perform poorly for Cox models.
} 
%\ylcm{the following is not smooth and the method was introduced very abruptly. you need to first introduce debiased lasso, then discuss the issues, then talk about the need to use quadratic programming} 
{Not imposing any sparsity conditions on the inverse information matrix,} we propose to estimate each row of  $\hTheta$ by  solving 
%\bincm{What is it? Do you want to say one-step estimator using the Taylor expansion? It is confusing without defining it here.} \lucm{I wanted to say one-step update from the lasso estimator. But to avoid confusion, I can change this to `` we consider a debiased lasso estimator by first solving..."} 
the following quadratic programming problem for $m$ ($j = 1, \ldots, p$):
\begin{equation} \label{eq:qp}
\min \{  m^T \hSigma m: m \in \mR^p, \| \hSigma m - e_j \|_{\infty} \le \gamma_n   \},
\end{equation}
where $\gamma_n \ge 0$ is a tuning parameter, $e_j$ is the vector with one at the $j$th element and zero elsewhere, and the $p \times p$ matrix
\begin{equation} \label{eq:hatsigma}
\hSigma =  n^{-1} \sum^n_{i=1} \delta_i \{ X_i - \heta_n(Y_i; \hbeta) \}^{\otimes 2}.
\end{equation}
{In the end, we obtain $\widehat{\Theta}$ as  a $p\times p$ matrix consisting of all $p$ solutions to \eqref{eq:qp} as its corresponding row vectors.}
%\ylcm{without a context, refs might be confused by use of $\hSigma$  vs  $\ddot{\ell}_n (\hbeta)$. addressing my comments above may help. }
{ 
%Unlike the methods for linear or generalized linear models,
Of note, we use $\hSigma$ in \eqref{eq:qp} in lieu of  $\ddot{\ell}_n (\hbeta)$, which is for  theoretical convenience that becomes evident in Section \ref{sec:theory}.
%$\hTheta$ is usually obtained as the approximated inverse of the second order derivative matrix of the loss. Although $\hSigma$ in \eqref{eq:qp} can be replaced by the second order derivative $\ddot{\ell}_n (\hbeta)$, we still use $\hSigma$ as in \eqref{eq:hatsigma} for the theoretical convenience that later becomes evident in Section \ref{sec:theory}.
}
In fact, under the assumptions in Section \ref{sec:theory}, we do have $\| \hSigma - \ddot{\ell}_n(\hbeta) \|_{\infty} = \oP(1)$ with a desirable convergence rate (see the proof of Theorem \ref{thm:main} in the Appendix), 
%\bincm{Where is it? I can't find it in Section 3. Do you also need a rate?} \lucm{Section 3 only contains the assumptions under which this convergence in probability to zero argument holds. The rate is similar to that in the proof of Theorem 1; see Eq. (A3), where $B_n$ can be replaced by $\ddot{\ell}_n(\hbeta)$. I can change this to ``This is because, under the assumptions stated in Theorem 1, Section 3, we may derive the rate $\| \hSigma - \ddot{\ell}_n(\hbeta) \|_{\infty} = \OP(s_0 {\lambda}_n)$". I feel like adding one or two sentences explaining how to arrive at this rate at the end of the proof of Theorem 1 in the Appendix is adequate.} 
and the numerical difference in the resulting debiased lasso estimators is negligible.  %Then the debiased lasso estimator for $\beta^0$ is 
%\begin{equation} %\label{eq:dslasso}
%\widehat{b} = (\widehat{b}_1, \ldots, \widehat{b}_p)^T = \hbeta - \widehat{\Theta} \dot{\ell}_n(\hbeta).
%\end{equation}

{Our approach extends  \citet{javanmard2014confidence}
%, which proposed
%a similar approach in 
in a linear regression setting to survival models.
%\citet{javanmard2014confidence}
%considered a similar  
%that is, solving (\ref{eq:qp}) } with $\widehat{\Sigma} = \sum_{i=1}^n X_i X_i^T / n$. {Indeed, the matrix $\hTheta$ obtained through \eqref{eq:qp} is non-sparse, a property shared by \citet{javanmard2014confidence}. 
However, %differing from \citet{javanmard2014confidence}, 
 as $\heta_n(Y_i; \hbeta)$ involves  all subjects,  $\widehat{\Sigma}$ given in (\ref{eq:hatsigma}) is no longer a sum of independent and identically distributed terms,  posing additional theoretical difficulties. We have addressed these challenges in our proofs.
 
Computationally, our proposed (\ref{eq:qp}) can be  implemented fairly fast for moderate dimensions and parallelized for high dimensions by using the R function \texttt{solve.QP}. Our simulations demonstrate its computational efficiency.}

% \ylcm{I polished the sentences, but logically the following is still not clear. Specifically, compared to Yu method, is non-sparse
% matrix an advantage (if so why?) ? or the computation is the advantage.}
% The debiased lasso procedure that solves (\ref{eq:qp}) is also related to \citet{yu2018confidence}. However, \citet{yu2018confidence} employed a {CLIME} estimator for obtaining the inverse information matrix estimation column-wise by solving a series of linear programming problems: 
%\[
%\min \{  \| m \|_1: m \in \mR^p, \| \hSigma m - e_j \|_{\infty} \le \gamma_n   \}, ~j=1, \ldots, p.
%\]
%As opposed to \citet{yu2018confidence}, our  method produces a non-sparse \ylcm{so what? cannot comment on the difference for the sake of difference. more importantly,
%does the difference lead to anything better??} $\widehat{\Theta}$ due to the use of a quadratic loss, a property shared with \citet{javanmard2014confidence}. \lucm{THIS PARAGRAPH OF THE DIFFERENCES BETWEEN CLIME AND QUANDRATIC PROGRAMMING METHOD WILL BE REMOVED!}
%which is the resulting $\widehat{\Theta}$ is not necessarily a sparse matrix due to the quadratic loss. 

\subsection{Selection of the tuning parameter}

Selecting a proper tuning parameter $\gamma_n$ is  critical for bias correction in $\widehat{b}$, which can be illustrated by a  simulation study. We simulate $n=500$ independent subjects, each with $p=100$ independent covariates generated from $N(0,1)$. Only two coefficients in $\beta^0$ in the Cox model are non-zero, taking values of 1 and 0.3. The underlying survival time $Y$ follows an exponential distribution with a rate of $\exp{(X^T \beta^0)}$, and the censoring time is simulated from an exponential distribution with a rate of $0.2\exp{(X^T \beta^0)}$, resulting in a censoring rate of about 20\%. Figure \ref{fig:tuning} depicts how the estimation bias and the empirical coverage probability from the debiased lasso approach change as $\gamma_n$ ranges from 0 to 1,   revealing that $\gamma_n$ within the shaded range would yield desirable inference results.

%The cross-validation criterion in \citet{cai2011constrained} for estimating the sparse inverse covariance matrix is inappropriate for our purpose, since it originates from the idea of maximizing the log likelihood for independent and identically distributed Gaussian random vectors. \citet{yu2018confidence} used 10-fold cross-validation to choose the $\gamma_n$ that minimizes the criterion $\mathrm{tr}(\diag(\widehat{\Sigma}\widehat{\Theta} - I_p)^2)$, which is an alternative option given in the R package \texttt{clime} but  still leaves large biases in the true signals in their simulation studies.  \citet{van2006cross} proposed a cross-validated partial log-likelihood criterion base on leave-one-out estimates, which has been modified for $K$-fold cross-validation and implemented in the R package \texttt{glmnet} for Cox model. However, our simulation shows that using this criterion this criterion tends to select the largest possible $\gamma_n$ and makes no bias correction, which is intuitive since the same criterion is adopted by \texttt{glmnet} to obtain the lasso estimator and facilitates relatively stable prediction rather than bias correction and inference.

\begin{figure}[ht]
	\centering
	\includegraphics[width=0.5\textwidth]{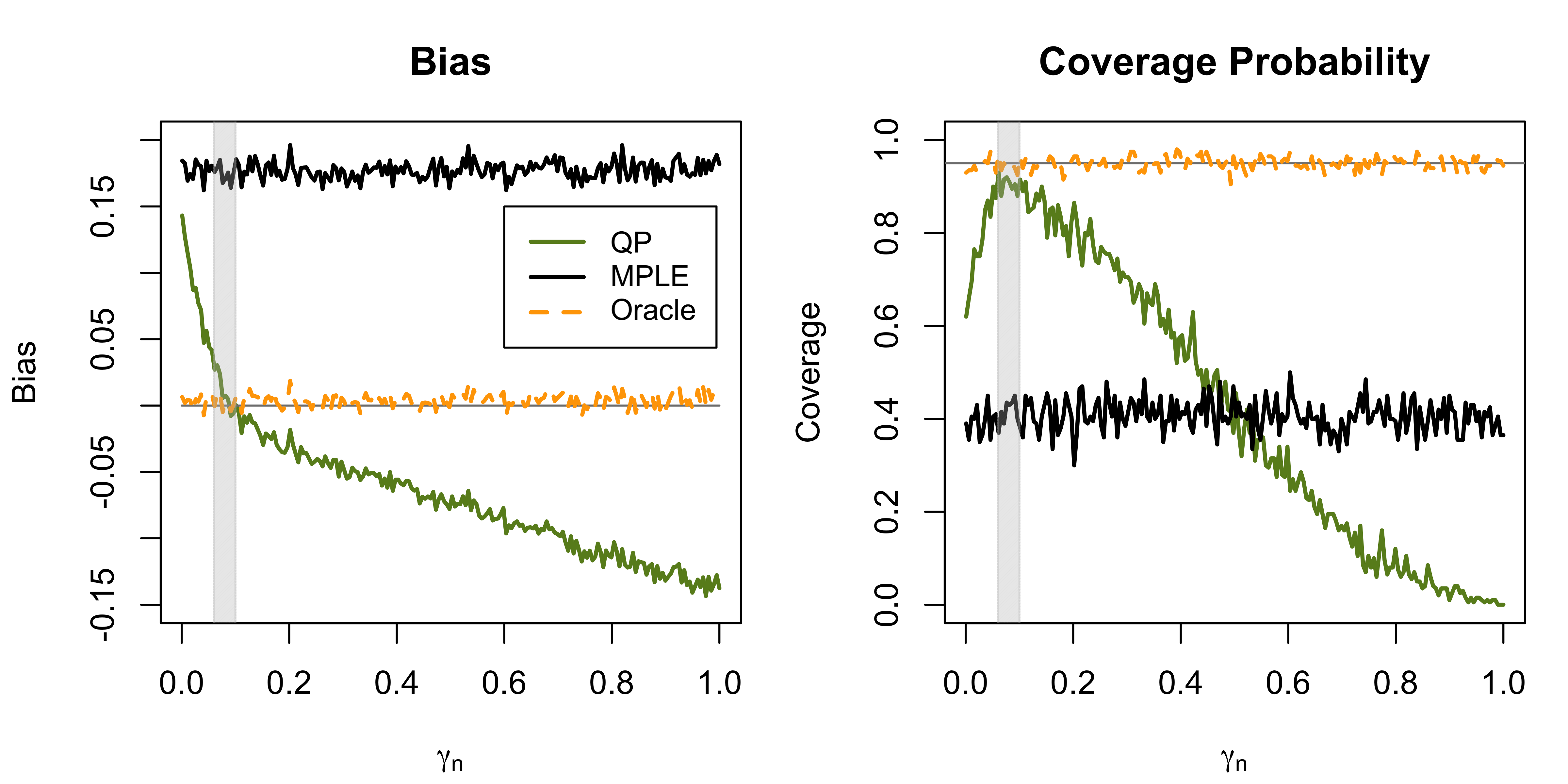}
	\caption{Estimation bias and 95\% confidence interval coverage probability for $\beta_1^0=1$ with the tuning parameter $\gamma_n \in [0,1]$ in a simulated example with $n=500$ observations and $p=100$ independent covariates. The methods in comparison include the proposed debiased lasso with quadratic programming (QP), the maximum partial likelihood estimation (MPLE) and the oracle estimator (Oracle) obtained from fitting the true model.}
	\label{fig:tuning}
\end{figure}

%To achieve sufficient bias correction and reliable inference, a desirable tuning parameter $\gamma_n$ should be close to 0 and results in a debiased estimator fitting the data well. 
We have found that, when evaluating  cross-validation criteria for choosing $\gamma_n$,  directly plugging in  debiased estimates produces highly unstable values because of accumulative errors from inclusion of the estimates for a large number of noise covariates. % thus discouraged. %as estimation error from every component of the debiased lasso estimator can accumulate and mounts to severe accuracy problems. 
Instead, we propose  a cross-validation procedure by hard-thresholding debiased estimates:  {splitting data randomly into $K$ folds  ($K=5$ or $10$), we  use the $k$th fold to obtain a debiased lasso estimate $\widehat{b}^{(k)}$,  hard-threshold it and plug in the thresholded values for computing cross-validation criteria.
%$\widehat{b}^{(k)}$ 
Hard-thresholding is based on multiple testing with, for example, the Bonferroni correction.  That is, we take the  hard-thresholded values} to be  $\widehat{b}^{(k),HT}_j=\widehat{b}^{(k)}_j$ if $\sqrt{n}|\widehat{b}^{(k)}_j|/\hTheta_{jj} > z_{\alpha/(2p)}$, or $0$ otherwise, where $z_{\alpha/(2p)}$ is the upper $(\alpha/(2p))$th percentile of $N(0,1)$, as determined by the asymptotic result given in Theorem \ref{thm:main}. Then, letting  $\ell^{(k)}$ be the negative log partial likelihood  [defined as in \eqref{eq:negloglik} but applied to the $k$th testing set]  evaluated at $\widehat{b}^{(k),HT}$,  we choose $\gamma_n$ that gives the smallest cross-validated negative partial likelihood, {$\sum_{k=1}^K n^{(k)} \ell^{(k)}$, where $n^{(k)}$ is the number of observations in the $k$th testing set}. {Use of an alternative cross-validated partial likelihood  \citep{verweij1993cross}
gives similar results.}

%\ylcm{did u ever define the   cross-validation criteria?} \lucm{Yes, we use the negative partial likelihood as cross-validation criteria.}
%This intuitive extra step of hard-thresholding helps reduce noise from the debiased estimator and is effective in practice.  

%\begin{figure}[ht]
%	\centering
%	\includegraphics[width=0.8\textwidth]{proposal_bias&cov_large_univCox_main5_n500p100_indep.png}
%	\caption{An illustrative simulation example for selecting the tuning parameter $\gamma_n$}
%	\label{fig:select_gamma}
%\end{figure}

%========================================================
\section{Theoretical results}
\label{sec:theory}

{We infer $c^T \beta^0$ for a loading vector $c \in \mR^p$ or $A \beta^0$ for a loading matrix $A \in \mR^{l \times p}$, by  studying the asymptotic properties for linear combinations of  $\widehat{b}$. }
Denote the expectation of $\hmu_r(t; \beta)$ as $\mu_r(t; \beta) = E [1(Y \ge t) X^{\otimes r} \exp\{ X^T \beta \}]$, and define population-level counterparts for $\heta_n(t; \beta)$ as
$
\eta_0(t; \beta) = \mu_1(t; \beta)/\mu_0(t; \beta),
%= \displaystyle \frac{E [1(Y \ge t) \exp\{ X^T \beta \} X]}{E [1(Y \ge t) \exp\{ X^T \beta \}]},
$
 and for $\hSigma$ in \eqref{eq:hatsigma} as
$
\Sigma_{\beta^0} = E \left[ \{ X - \eta_0(Y; \beta^0) \}^{\otimes 2} \delta  \right].
$
Denote by $\Theta_{\beta^0} = \Sigma_{\beta^0}^{-1}$.  {We enumerate sufficient conditions needed for} establishing the theoretical properties of the debiased lasso estimator.
\begin{itemize} %\setlength{\itemindent}{15pt}
	\item[] \textit{Assumption 1}. Covariates are almost surely uniformly bounded, i.e. $\| X_i \|_{\infty} \le K$ for some constant $K<\infty$ for $i = 1,2, \ldots, n$.
	\item[] \textit{Assumption 2}. $| X_i^T \beta^0 | \le K_1$ uniformly for all $i = 1, \cdots, n$ with some constant $K_1 < \infty$ almost surely.
	\item[] \textit{Assumption 3}. The follow-up time stops at a finite time point $\tau > 0$, where the probability $\pi_0 = \PP (Y \ge \tau) > 0$.
	\item[] \textit{Assumption 4}. Let 
	\[
	\widetilde{\Sigma}_{\beta^0}(t) = \int_0^t  \left\{ \mu_2(u; \beta^0) - \frac{\mu_1(u; \beta^0) \mu_1^T(u; \beta^0)}{\mu_0(u; \beta^0)} \right\} dH_0(u).
	\]
	For any $t\in [0, \tau]$, we assume 
	\[
	\frac{c^T \Theta_{\beta^0} \widetilde{\Sigma}_{\beta^0}(t)  \Theta_{\beta^0} c}{c^T \Theta_{\beta^0} c}   \rightarrow v(t; c), ~ \mathrm{as} ~ n\rightarrow \infty
	\]
	for some fixed function $v(\cdot; c) > 0$.
	\item[] \textit{Assumption 5}. The matrix $\Sigma_{\beta^0}$ has bounded eigenvalues, i.e. there exist two constants $\zeta_\mathrm{min}$ and $\zeta_\mathrm{max}$ such that $0 < \zeta_\mathrm{min} \le \zeta_{\mathrm{min}}(\Sigma_{\beta^0}) \le \zeta_{\mathrm{max}}(\Sigma_{\beta^0}) \le \zeta_\mathrm{max} < \infty$, where $ \zeta_{\mathrm{min}}(\Sigma_{\beta^0})$ and $ \zeta_{\mathrm{max}}(\Sigma_{\beta^0}) $ represent the smallest and the largest eigenvalues of $\Sigma_{\beta^0}$.
	
\end{itemize}

It is common in the literature of high-dimensional inference to assume bounded covariates as in \textit{Assumption 1}. \citet{fang2017testing} and \citet{kong2018high} also posed \textit{Assumption 2} for the Cox model inference, i.e. uniform boundedness on the multiplicative hazard. Under \textit{Assumption 1}, \textit{Assumption 2} can be implied by bounded overall signal $\| \beta^0 \|_1$. \textit{Assumption 3} is usually used for survival models with censored data \citep{andersen1982cox}. \textit{Assumption 4}  ensures the convergence of a predictable variation process in the Martingale central limit theorem and thus the asymptotic normality of the de-biased lasso estimator. $\widetilde{\Sigma}_{\beta^0}(t)$ can be viewed as the information matrix up to time point $t$. It is easy to see that $\widetilde{\Sigma}_{\beta^0}(\tau)=\Sigma_{\beta^0}$ and $v(\tau; c)=1$. {This assumption states that the  limiting function $v(t; c)$ also depends on $c \in \mR^p$,  the  loading vector of interest, which is reasonable}. 
%\textit{Assumption 4}  is an alternative assumption to the stringent boundedness condition on $\| \Theta_{\beta^0} X_i \|_{\infty}$, which was essential in \citet{van2014asymptotically} for statistical inference in high-dimensional generalized linear models and in \citet{fang2017testing} for the Cox model.  
%\bincm{I hided the above statement - see latex file - because I don't see direct connection between the two assumptions} 
The bounded eigenvalue condition on $\Sigma_{\beta^0}$ in \textit{Assumption 5} is standard in inference for high-dimensional models. 
% YL: I did not see the connection below. 
%Since we focus on random designs following \citet{kong2014non}, unlike \citet{huang2013oracle}, \citet{fang2017testing} and \citet{yu2018confidence}, we do not directly assume the compatability condition on $\ddot{\ell}_n(\beta^0)$.

%\textcolor{blue}{Some comments: \textit{Assumption} 4 can be implied if $\| \Theta_{\beta^0} \|_{1,1}  = \mathcal{O}(1)$ for bounded design. If $\Sigma_{\beta^0} = \Gamma \Lambda \Gamma^T$ in the eigen decomposition, then $\Theta_{\beta^0} X_i = \Gamma \Lambda^{-1} \Gamma^T X_i$ (rotation, expansion and contraction along the axes, then rotation; it $\ell_{\infty}$ norm might be still relevant to $p$).}

%The following theorem establishes the asymptotic distribution for linear combinations of the resulting debiased lasso estimator $\widehat{b}$.

\begin{theorem} \label{thm:main}
	Assume that the two tuning parameters satisfy $\lambda_n \asymp \sqrt{\log(p)/n}$ and $\gamma_n \asymp \| \Thetabeta \|_{1,1} s_0 \lambda_n$. Furthermore, assume $\| \Theta_{\beta^0} \|_{1,1}^2 p s_0 \log(p) / \sqrt{n} \rightarrow 0$ as $n \rightarrow \infty$. Under  Assumptions 1--5, for any $c \in \mathbb{R}^p$ such that $\| c \|_2 = 1$ and $\| c \|_1 \le a_*$ with some absolute constant $a_* < \infty$, we have
	\[
	\sqrt{n} c^T (\widehat {b} - \beta^0) / (c^T \widehat{\Theta} c)^{1/2} \overset{\mathcal{D}}{\rightarrow} N(0,1).
	\]
\end{theorem}

%\bincm{I don't see where you define $\asymp$. Should you define it at the beginning of Section 2.1?} \lucm{Yes, added this definition at the beginning of Section 2.1.}
Theorem \ref{thm:main} provides the foundation for drawing inference on the regression coefficients. In the following,  Corollary \ref{coro:power}(i) discusses the type I error and the power of testing $H_{0}: c^T \beta^0 = a_0$ based on Theorem \ref{thm:main}, and Corollary \ref{coro:power}(ii) ensures that the corresponding confidence interval achieve nominal coverage probability asymptotically.

\begin{corollary} \label{coro:power}
	Suppose that the assumptions in Theorem \ref{thm:main} hold. 
	
    (i) {To test  a null  hypothesis  $H_{0}: c^T \beta^0 = a_0$ 
    versus an alternative hypothesis $H_{1}: c^T \beta^0 = a_1$, where $a_1 \ne a_0$, 
    with a} known $c \in \mR^p$ and constant $a_0 \in \mR$, let the test statistic $T = \sqrt{n} (c^T \widehat{b} - a_0) / (c^T \widehat{\Theta} c)^{1/2}$. We construct a test function 
\[
\phi(T) = \left\{  \begin{array}{ll}
1 & \quad \mathrm{if} \ |T| > z_{\alpha/2} \\
0 & \quad \mathrm{if} \ |T| \le z_{\alpha/2}
\end{array} ,  \right.
\]
where $z_{\alpha/2}$ is the upper $(\alpha/2)$th quantile of $N(0,1)$. Then, the type I error rate for the test $\phi(T)$ satisfies $\PP (\phi(T)=1 | H_0) \to \alpha$, and the power under the alternative $H_1$ satisfies $\PP (\phi(T) = 1 | H_1) \to 1$ as $n \to \infty$. 

(ii) The two-sided level $\alpha$ confidence interval for $c^T \beta^0$ can be  constructed as $CI(\alpha) = [ c^T \widehat{b} - z_{\alpha/2} (c^T \widehat{\Theta} c / n)^{1/2},~ c^T \widehat{b} + z_{\alpha/2} (c^T \widehat{\Theta} c / n)^{1/2}]$. Then {$\PP (c^T \beta^0 \in CI(\alpha) | H_0) \rightarrow 1 - \alpha$ } as $n \rightarrow \infty$. %\ylcm{make sure the proof has the correct notation.}
\end{corollary}

With Theorem \ref{thm:main} and the Cram\'er-Wold device, we can also conduct simultaneous inference on multiple linear combinations, i.e. $A \beta^0$ for some $l \times p$ matrix $A$, as summarized in the following Theorem \ref{thm:simul}, with \textit{Assumption 4} replaced by its multivariate version, \textit{Assumption 6}. Similarly, Corollary \ref{coro:thm2} provides the asymptotic results for hypothesis testing and confidence region in this setting.

\begin{itemize}
    \item[] \textit{Assumption 6.} Let $\widetilde{\Sigma}_{\beta^0}(t)$ be the same as in \textit{Assumption 4}. For a fixed combination matrix of interest $A \in \mR^{l \times p}$, it holds that
	\[
	\frac{\omega^T A \Theta_{\beta^0} \widetilde{\Sigma}_{\beta^0}(t)  \Theta_{\beta^0} A^T \omega}{\omega^T A \Theta_{\beta^0} A^T \omega}   \rightarrow v^{\prime}(t; A^T \omega), ~ \mathrm{as} ~ n\rightarrow \infty
	\]
	for any vector $\omega \in \mR^l$ and any $t\in [0, \tau]$, where $v^{\prime}(\cdot; A^T \omega) > 0$ is some fixed function depending on $A^T \omega$.
\end{itemize}

\begin{theorem} \label{thm:simul}
	Let $A$ be an $l \times p$ matrix of full row rank such that the number of rows $l$ is fixed, $\| A \|_{\infty, \infty} = \mathcal{O}(1)$ and $A \Theta_{\beta^0} A^T \rightarrow F$ for some fixed $l \times l$ matrix $F$. Assume that the two tuning parameters $\lambda_n \asymp \sqrt{\log(p)/n}$ and $\gamma_n \asymp \| \Thetabeta \|_{1,1} s_0 \lambda_n$, and that $\| \Theta_{\beta^0} \|_{1,1}^2 p s_0 \log(p) / \sqrt{n} \rightarrow 0$ as $n \rightarrow \infty$. Under Assumptions 1--3, 5 and 6, we have
	\[
	\sqrt{n} A  (\widehat{b} - \beta^0) \overset{\mathcal{D}}{\rightarrow} N (0, F).
	\]
\end{theorem}

\begin{corollary} \label{coro:thm2}
	Suppose the assumptions in Theorem \ref{thm:simul} hold. 
	
	(i) For the $l \times p$ matrix $A$ in Theorem \ref{thm:simul}, under the null hypothesis $H_0: A \beta^0 = a_0$ for some $a_0 \in \mR^l$,  the statistic  $T^{\prime} = n (A \widehat{b} - a^0)^T \widehat{F}^{-1} (A \widehat{b} - a^0)  \overset{\mathcal{D}}{\rightarrow} \chi^2_l$, where $\widehat{F} = A \hTheta A^T$. 
	
	(ii) For $\alpha \in (0,1)$, let the confidence region for $A \beta^0$ be $CR(\alpha) = \{ a \in \mR^l: n (A \widehat{b} - a)^T \widehat{F}^{-1} (A \widehat{b} - a) \le \chi^2_{l, \alpha} \}$, where $\chi^2_{l, \alpha}$ is the upper $\alpha$th percentile from $\chi^2_l$. Then 
	{$\PP (A\beta^0 \in CR(\alpha) | H_0) \rightarrow 1 - \alpha$ as $n \rightarrow \infty$.}
\end{corollary}

Proofs of Theorems \ref{thm:main} and \ref{thm:simul} are provided in the Appendix. Corollaries \ref{coro:power} and \ref{coro:thm2} are directly obtained from Theorems \ref{thm:main} and \ref{thm:simul}, and their proofs are omitted.

%=====================================================
\section{Numerical experiments}
\label{sec:simulation}

{For a total of $n=500$ subjects, we simulate $p=20,100,200$ covariates, respectively, and generate these covariates from $N(0, \Sigma)$, where $\Sigma=I_p$ and \textsc{AR(1)} with the  correlation parameter of 0.5 %\bincm{In the figure you said 0.7. Which is correct?} \lucm{$\rho=0.5$ is correct, after I checked both my dissertation and the saved simulation results. $\rho=0.7$ in the table must be some typo passed down from the older draft.},  
as two different setups.  Each covariate is truncated at $\pm 2.5$.}
%, ~ j=1,\ldots, p$. }
%\bincm{Have you defined $X^{(j)}$ before?} \lucm{No, but already added the definition at the beginning of Section 2.1.} 
{Concerning the specifications of the true regression coefficients $\beta^0$, the first element $\beta_1^0$ varies from 0 to 2 with an equal step size  of 0.2, four of the other elements are arbitrarily chosen to take values of 1, 1, 0.5 and 0.5, and the rest are set to be zero. The underlying survival times $T$ and the censoring times $C$ are independently generated from an exponential distribution with hazard  $h(t|X) = \exp\{X^T\beta^0\}$, and from $\mathrm{Uniform}(1,20)$, respectively. Under each simulation configuration, 200 datasets are generated.}

The methods in comparison include: (1) QP: our proposed debiased lasso with quadratic programming for matrix $\widehat{\Theta}$; (2) NW: the debiased lasso with node-wise lasso for matrix $\widehat{\Theta}$  in \citet{kong2018high}; (3) CLIME: debiased lasso with CLIME for matrix $\widehat{\Theta}$  in \citet{yu2018confidence}; (4) Decor:  decorrelated Wald test in \citet{fang2017testing} and (5) Oracle: the estimator when the true model is known a priori.  

%We compare the the estimation bias for $\beta^0_1$, its model-based standard error, coverage probability at significance level $\alpha=0.05$ and mean squared error.  

%\yl{To select tune parameters for QP,  we utilize the selection procedure described in Section 2.3.}
%the proposed debiased lasso estimator via quadratic programming (QP). 

For the lasso estimator, we use 10-fold cross-validation to select the tuning parameter $\lambda_n$. Five-fold cross-validation is used for tuning parameter selection in {CLIME}, QP and NW. For the hard-thresholding step used to select $\gamma_n$ as described in Section 2.3, we adopt the Bonferroni correction with the adjusted p-value threshold $0.1/p$, where $p$ is the number of covariates.

%Except for \textit{Decor-W},  the other methods has a good control of the type-1 error rates (Table). Figure  shows that with the same tuning parameter selection procedure, \textit{DS-QP} and \textit{DS-CLIME} generate similar results. They greatly improve the bias correction behavior of the debiased lasso, compared to \textit{DS-NW}, and the corresponding 95\% CIs achieve empirical coverage probability close to 90\%, the highest among all. 

We compare these methods with respect to the bias of the estimated $\beta^0_1$ (the parameter of main interest), its model-based standard error, coverage probability with a significance level of $\alpha=0.05$ and mean squared error.  Figures~\ref{fig:indep} and \ref{fig:ar1} show the results for the independent and the AR(1) covariance structures, respectively. When  $p=20$, our proposed method (QP) and the decorrelated Wald test (Decor) {perform nearly as well as} the oracle estimator (Oracle) and {MPLE}. When the dimension is relatively large compared to the sample size, i.e. $p=100, 200$, {next to Oracle}, the proposed estimator (QP) displays the smallest biases and the confidence intervals with coverage probabilities closest to the nominal level 95\% for both covariance structures. {On the other hand, NW, CLIME, Decor and MPLE incur substantial biases as the true value of $\beta^0$ increases.
%\lucm{deleted ``owing to estimating a non-sparse $\Theta_{\beta^0}$ as a sparse matrix"}. 
In addition, owing to the estimation of $\Theta_{\beta^0}$ using penalized approaches, the model-based standard error estimates using  NW and CLIME are shrunk towards zero, underestimating the true variation. As such, the four competing methods all present improper confidence interval coverage probabilities, whereas our proposed method retains nearly unbiased estimates with coverage probabilities close to the nominal level}.
%Compared to the independent covariance case, the proposed method (QP) performs worse in the AR(1) covariance case in terms of bias correction and confidence interval coverage, but is still the best in all methods considered, especially when $p$ is relatively large. 
%\ylcm{please comment on the poor performance of mple}
\begin{figure}[!htb]
	\centering
	\includegraphics[width=12cm]{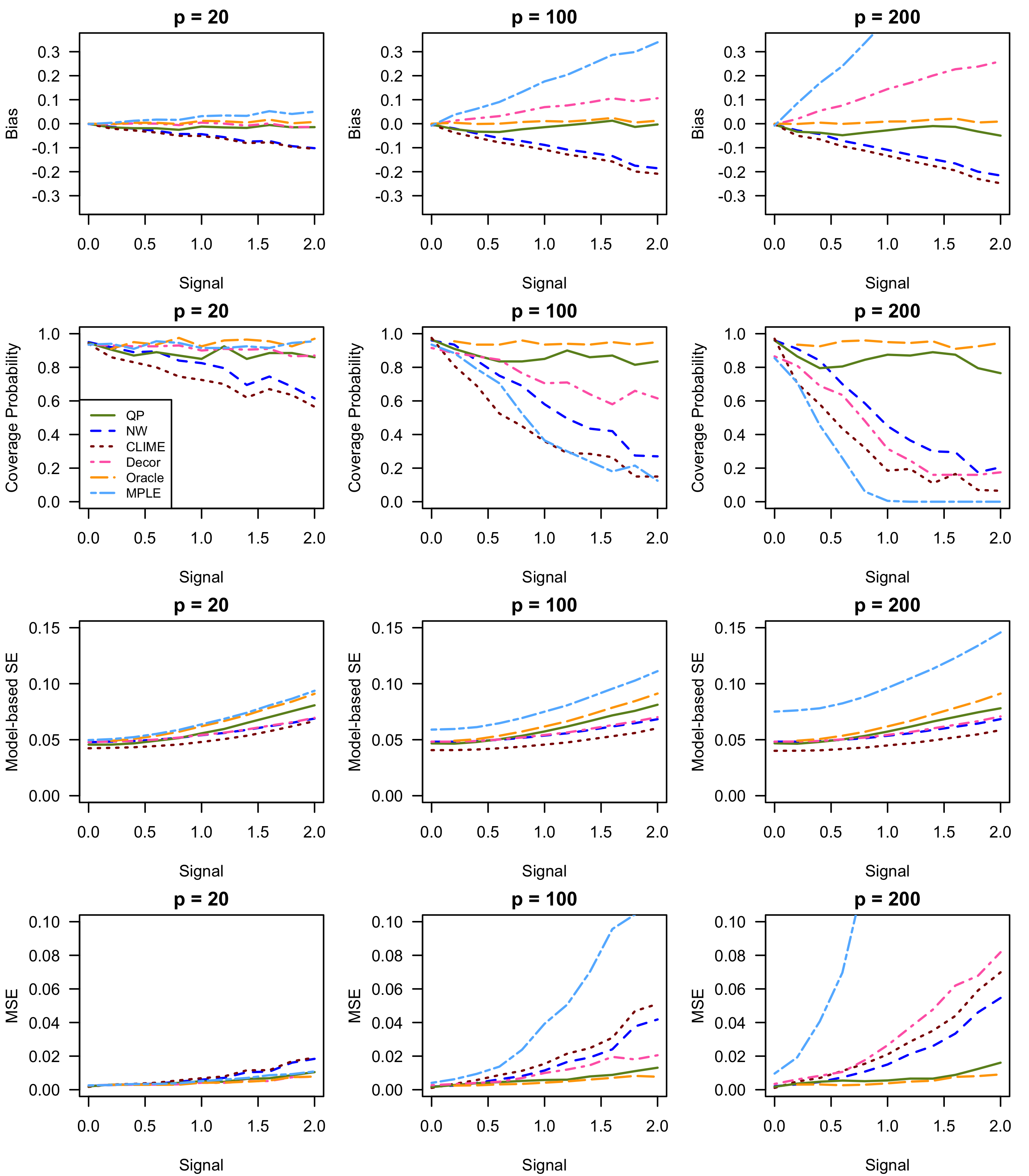}
	\caption{Estimation bias, coverage probability, model-based standard error and mean squared error for six estimators in comparison, QP (solid green lines), NW (short-dash navy blue lines), CLIME (dotted red lines), Decor (dot-dash pink lines), Oracle (long-dash orange lines), and MPLE (two-dash light blue lines), based on 200 simulations, each with $n=500$ observations and independent covariance structure for covariates. 
		\label{fig:indep}}
\end{figure}

\begin{figure}[!htb]
	\centering
	\includegraphics[width=12cm]{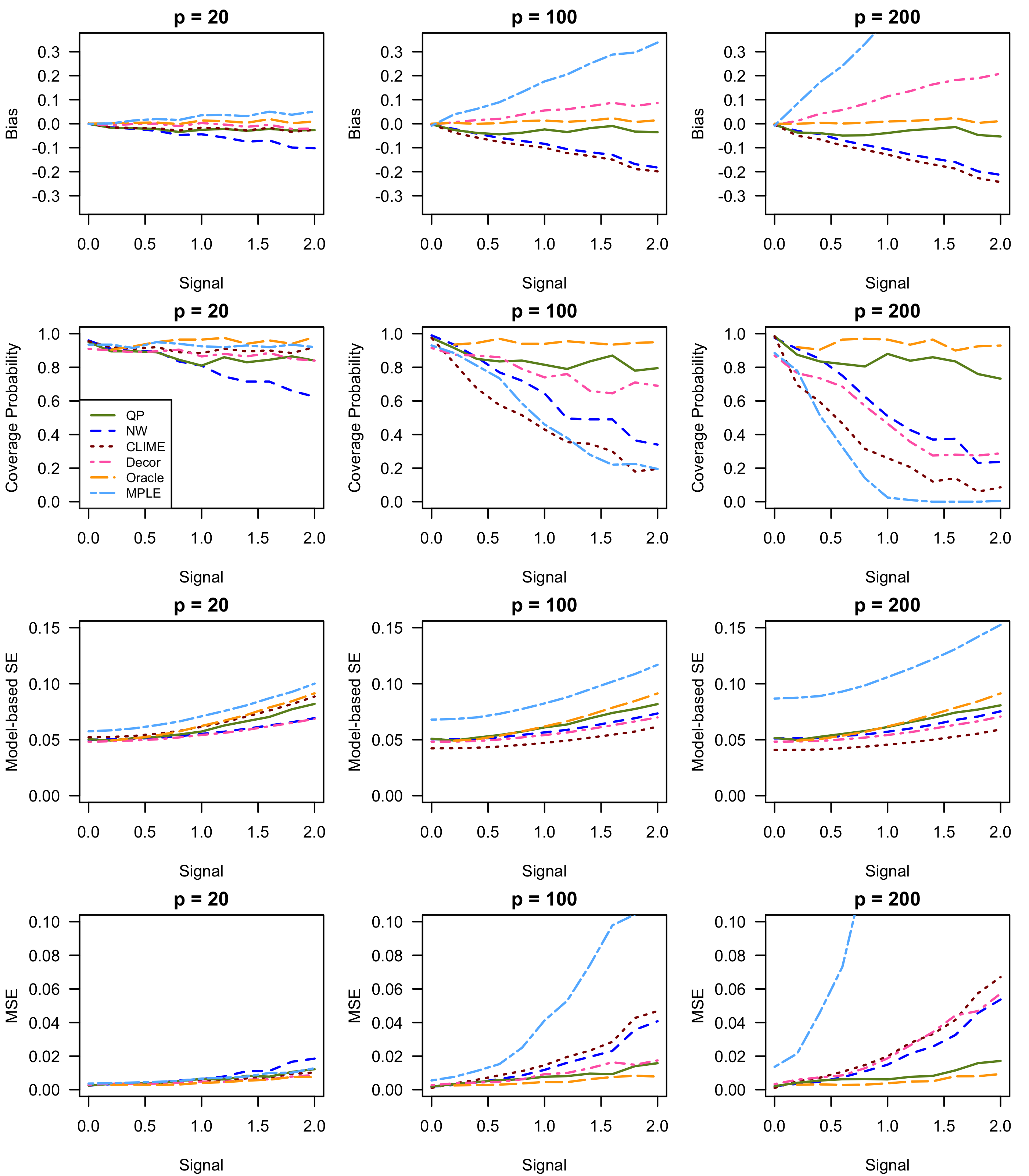}
	\caption{Estimation bias, coverage probability, model-based standard error and mean squared error for six estimators in comparison, QP (solid green lines), NW (short-dash navy blue lines), CLIME (dotted red lines), Decor (dot-dash pink lines), Oracle (long-dash orange lines), and MPLE (two-dash light blue lines), based on 200 simulations, each with $n=500$ observations and AR(1) covariance structure for covariates ($\rho = 0.5$). 
		\label{fig:ar1}}
\end{figure}

{We next compare the time spent on computing $\hTheta$ alone (Table \ref{tab:time}) among}  \texttt{solve.QP} in the R package \texttt{quadprog} for the proposed quadratic programming procedure, {and two commonly used \texttt{R} functions for CLIME, namely, \texttt{clime} in the package \texttt{clime} and \texttt{sugm} in the package \texttt{flare}}. %All data are included without cross-validation, and
Three candidate values of $\gamma_n$, namely, 0.3, 1 and 2 times of $\sqrt{\log(p)/n}$, are used for demonstration. We fix $\beta^0_1=1$ and simulate $n=500$ observations, {with covariates having an AR(1) covariance structure and the rest of the setting being identical to what is described in the first paragraph of this section. The time columns  in Table \ref{tab:time} report the average computing time over 10 replications  on a MacBook with 2.7GHz Intel Core i5 processor and 8GB memory, and  the  ratio columns compare the average computing time of each programming procedure to that of \texttt{solve.QP} for each simulation setting, respectively. Under all of the scenarios examined,} our proposed implementation with \texttt{solve.QP} is the most computationally efficient;  for  large dimensions, e.g., $p=200$,  \texttt{clime} takes the longest time per dataset on average.
%\ylcm{what is
%difference between clime and flare? do we need to say anything about flare?} \lucm{clime and flare are two different packages implementing CLIME estimation. flare includes other methods besides CLIME, but in terms of CLIME itself, the main difference between these two packages is on whether implementation is optimized, e.g. flare is generally faster for large dimensions.}

\begin{table}[ht]
	\centering
	\caption{Comparison of the computational time spent on computing $\hTheta$. Time (in seconds) is averaged over 10 replications under each setting. Time ratio is with respect to the proposed method implemented using \texttt{solve.QP}.}
	\label{tab:time}
	\begin{tabular}{lllllll}
		\hline
		%\headrow
		 & \multicolumn{2}{c}{{solve.QP}} & \multicolumn{2}{c}{{clime}} & \multicolumn{2}{c}{{flare}} \\
		\cmidrule(lr){2-3} \cmidrule(lr){4-5} \cmidrule(l){6-7}
		{$p=20$} & Time & Ratio & Time & Ratio & Time & Ratio  \\ 
		\hline
		$\gamma_n = 0.3\sqrt{\log(p)/n}$ & 0.0016 & 1.0 & 0.0392 & 24.5 & 0.1898 & 118.6 \\ 
		$\gamma_n = \sqrt{\log(p)/n}$ & 0.0015 & 1.0 & 0.0373 & 24.9 & 0.1597 & 106.5 \\ 
		$\gamma_n = 2\sqrt{\log(p)/n}$ & 0.0012 & 1.0 & 0.0338 & 28.2 & 0.1522 & 126.8 \\ 
		\hline
		% &  \multicolumn{2}{c}{{solve.QP}} & \multicolumn{2}{c}{{clime}} & \multicolumn{2}{c}{{flare}} \\
		%\cmidrule(lr){2-3} \cmidrule(lr){4-5} \cmidrule(l){6-7}
		{$p=100$} & Time & Ratio & Time & Ratio & Time & Ratio  \\ 
		\hline
		$\gamma_n = 0.3\sqrt{\log(p)/n}$ & 0.3159 & 1.0 & 4.3452 & 13.8 & 5.8860 & 18.6 \\ 
		$\gamma_n = 1\sqrt{\log(p)/n}$ & 0.0922 & 1.0 & 3.4164 & 37.1 & 2.0754 & 22.5 \\ 
		$\gamma_n = 2\sqrt{\log(p)/n}$ & 0.0665 & 1.0 & 2.6281 & 39.5 & 0.3663 & 5.5 \\ 
		\hline
		% &  \multicolumn{2}{c}{{solve.QP}} & \multicolumn{2}{c}{{clime}} & \multicolumn{2}{c}{{flare}} \\
		%\cmidrule(lr){2-3} \cmidrule(lr){4-5} \cmidrule(l){6-7}
		{$p=200$} & Time & Ratio & Time & Ratio & Time & Ratio  \\ 
		\hline
		$\gamma_n = 0.3\sqrt{\log(p)/n}$ & 4.3886 & 1.0 & 64.7047 & 14.7 & 52.2224 & 11.9 \\ 
		$\gamma_n = 1\sqrt{\log(p)/n}$ & 0.9039 & 1.0 & 47.0320 & 52.0 & 21.7229 & 24.0 \\ 
		$\gamma_n = 2\sqrt{\log(p)/n}$ & 0.6196 & 1.0 & 33.0308 & 53.3 & 2.5536 & 4.1 \\ 
		\hline
	\end{tabular}
\end{table}

%===============================================
\section{Boston lung cancer data analysis}
\label{sec:app}

Lung cancer is the leading cause of cancer deaths in the United States, and non-small cell lung cancer (NSCLC), accounting for approximately 80\% to 85\% among all the  lung cancer cases, is the most common histological type of lung cancer \citep{houston2018histologic}.  Identification of genetic variants associated with lung cancer patient survival sparks modern translational cancer research, and has the potential to refine prognosis  and promote individualized treatment and clinical care. Despite numerous studies investigating  potential predisposing genes to lung cancer risks, studies on patient survival usually have small sample sizes and the reported genetic markers associated with lung cancer survival have been poorly replicated \citep{bosse2018decade}. The Boston Lung Cancer Survival Cohort (BLCSC) is a large epidemiology cohort for investigating the molecular cause underlying lung cancer, where lung cancer cases have been enrolled at Massachusetts General Hospital and the Dana-Farber Cancer Institute from 1992 to present. We apply the proposed debiased lasso method (QP) to 
%a subset \bincm{What subset? Can you define it?} \lucm{I said we used a subset because BLCSC is a very large study cohort. The subset is defined right in the next paragraph. Should we add a note like ``we apply ... to a subset of the BLCSC data as defined below"? Otherwise, I can move the first sentence in the next paragraph here (with slight changes).} of 
a BLCSC cohort with genetic data and simultaneously investigate the joint effects of certain genotyped SNPs on NSCLC patient overall survival.

%The 
%subset of 
Included in the analysis are $n=561$ NSCLC patients with available diagnosis dates, follow-up times and genotypes on Axiom arrays. Among all these patients, 437 (77.9\%) died and 124 (22.1\%) were censored. 
The range of the observed survival time is from 6 days to $8584$ days, and the restricted mean survival and censoring times at $\tau=8584$ days are 2124 (SE: 105) and 4397 (SE: 187) days, respectively. Patient characteristics,
%adjusted in the Cox proportional hazards model, 
including age at diagnosis, race, education level, gender, smoking status, histological type, cancer stage, and treatment received, are provided in the online supplementary materials. % in Table \ref{tab:surv_pop}. 

A conventional marginal association analysis \citep{tang2020novel} found two potentially functional SNPs in the genes \textit{HDAC2} and \textit{PPARGC1A} that were significantly associated with NSCLC overall survival. Using the target gene approach, we focus on 32 genes in the CARM ER pathway, which is the largest pathway \citet{tang2020novel} considered and described in their supplementary document and contains the two reported genes \textit{HDAC2} and \textit{PPARGC1A}, plus 9 genes that \citet{xia2020revisit} studied to investigate whether the susceptibility loci are also associated with patient survival.  We extract 312 genotyped SNPs from the 32 genes in the CARM ER pathway and the nine target genes described in \citet{xia2020revisit} from the BLCSC data (minor allele frequency $>$ 0.01, genotype call rate $>$ 95\%). After a pruning step using PLINK \citep{purcell2007plink} to avoid multicolinearity caused by SNPs with high linkage disequilibrium, %(window size 50, step size 5, and $r^2 > 0.7$), 
the number of SNPs is reduced to 217. SNPs are coded by the number of copies of the minor allele, i.e. 0, 1 or 2, and assumed  to have additive effects on the log hazard ratio. Therefore, the subset of the BLCSC data we analyze include $n=561$ NSCLC patients and $p=231$ covariates.

Table \ref{tab:est2} summarizes the coefficient estimates in the Cox proportional hazards model for all patient characteristics and the top ten SNPs ranked by the p-values from the proposed method (QP). Results of two methods, QP versus MPLE, are listed side by side.  In general, QP results in points estimates of smaller magnitudes and smaller standard errors compared to MPLE, which is consistent with our observation in the simulated example. %\bincm{MPLE is not included in the simulation. So this sentence is not precise. Also should mention why MPLE is compared here -- a common approach to use? Is the acronym MPLE defined before?} \lucm{(1) Previously the acronym MPLE was only defined in Figure \ref{fig:tuning}, but I added it at the beginning of Section 1. (2)  MPLE appeared once in the simulation, i.e. Figure \ref{fig:tuning}, where MPLE is clearly biased and has poor coverage. (3) I might have saved the MPLE in the simulations represented in Section 4. Will check and add to the results there.} 
MPLE is numerically very unstable when the dimension $p$ is large compared to the sample size $n$. The numerical instability arises primarily from inverting the Hessian matrix, which may be closer to being singular. On the contrary, Lasso provides a more stabilized initial estimator. As a result, the debiased lasso estimator is numerically more stable than MPLE with narrower confidence intervals.
%since the standard errors are not estimated using the inverted Hessian matrix.  
When the dimension $p$ is very small, the difference between the two methods becomes negligible. %\bincm{Do you have supporting evidence to claim this? Again, should you include simulations for MPLE?} \lucm{I might have saved the MPLE in the simulations represented in Section 4. Will check and add to the results there.}

\begin{landscape}
	\begin{table}[h!]
		\centering
		\caption{Coefficient estimates in the Cox proportional hazards model for the Boston Lung Cancer Study data}
		\footnotesize
		\begin{threeparttable}	
			\begin{tabular}{lllllllllll}
			    \hline
				%\headrow %\thead{Variables}
				&  & & \multicolumn{4}{c}{QP}  & \multicolumn{4}{c}{MPLE}  \\
				\cmidrule(lr){4-7}  \cmidrule(lr){8-11}
				{Variable} & {{Note}}  & &  {Est} & {SE} & {P-value}  & {95\% CI} &  {Est} & {SE} & {P-value} & {95\% CI} \\
				\hline
				Race & \multicolumn{2}{l}{Others vs Caucasian} & -0.163 & 0.201 & 0.416 & (-0.557, 0.231) & 0.065 & 0.561 & 0.908 & (-1.034, 1.163) \\ 
				Education & \multicolumn{2}{l}{HS vs No HS} & -0.018 & 0.091 & 0.840 & (-0.198, 0.161) & -0.142 & 0.253 & 0.574 & (-0.637, 0.353) \\ 
				& \multicolumn{2}{l}{College vs No HS} & -0.037 & 0.076 & 0.625 & (-0.185, 0.111) & -0.085 & 0.218 & 0.698 & (-0.513, 0.343) \\ 
				Gender & \multicolumn{2}{l}{Male vs Female} & 0.314 & 0.075 & $<0.001$ & (0.166, 0.461) & 0.439 & 0.166 & 0.008 & (0.114, 0.763) \\ 
				Age & \multicolumn{2}{l}{Standardized} & 0.155 & 0.038 & $<0.001$ & (0.081, 0.230) & 0.400 & 0.090 & $<0.001$ & (0.224, 0.577) \\ 
				Smoker & \multicolumn{2}{l}{Yes vs No} & 0.103 & 0.142 & 0.470 & (-0.176, 0.381) & 0.066 & 0.299 & 0.825 & (-0.519, 0.651) \\ 
				Histology & \multicolumn{2}{l}{AD vs LCC} & -0.259 & 0.076 & 0.001 & (-0.409, -0.11) & -0.467 & 0.294 & 0.112 & (-1.043, 0.109) \\ 
				& \multicolumn{2}{l}{SCC vs LCC} & 0.065 & 0.094 & 0.488 & (-0.120, 0.251) & -0.030 & 0.314 & 0.923 & (-0.646, 0.585) \\ 
				& \multicolumn{2}{l}{Unspecified vs LCC} & 0.046 & 0.132 & 0.729 & (-0.213, 0.304) & -0.119 & 0.384 & 0.756 & (-0.871, 0.633) \\ 
				Stage & \multicolumn{2}{l}{Late vs Early} & 0.352 & 0.081 & $<0.001$ & (0.193, 0.510) & 0.553 & 0.190 & 0.004 & (0.180, 0.926) \\ 
				Surgery & \multicolumn{2}{l}{Yes vs No} & -1.102 & 0.085 & $<0.001$ & (-1.269, -0.936) & -2.115 & 0.226 & $<0.001$ & (-2.557, -1.672) \\ 
				Chemotherapy & \multicolumn{2}{l}{Yes vs No} & 0.025 & 0.078 & 0.753 & (-0.128, 0.177) & -0.239 & 0.220 & 0.278 & (-0.671, 0.193) \\ 
				Radiation & \multicolumn{2}{l}{Yes vs No} & 0.047 & 0.077 & 0.548 & (-0.105, 0.198) & 0.248 & 0.198 & 0.211 & (-0.140, 0.636) \\ 
				Treatment record & \multicolumn{2}{l}{Missing vs Not} & 0.099 & 0.176 & 0.573 & (-0.245, 0.443) & 0.347 & 0.428 & 0.417 & (-0.492, 1.186) \\ 
				\hline
				
				%\headrow
				{SNP} & {Pos} & {Gene} & {Est} & {SE} & {P-value} & {95\% CI} &  {Est} & {SE} & {P-value}  & {95\% CI}  \\ 
				\hline
				AX-11672686 & 8:27324822 & \textit{CHRNA2} & 0.186 & 0.054 & 0.001 & (0.081, 0.291) & 0.185 & 0.402 & 0.645 & (-0.603, 0.973) \\ 
				AX-11673610 & 12:66762242 & \textit{GRIP1} & 0.313 & 0.092 & 0.001 & (0.133, 0.494) & 0.773 & 0.220 & $<0.001$ & (0.343, 1.203) \\ 
				AX-11264571 & 13:32906729 & \textit{BRCA2} & 0.206 & 0.061 & 0.001 & (0.086, 0.325) & 0.450 & 0.164 & 0.006 & (0.129, 0.772) \\ 
				AX-40031129 & 16:3860539 & \textit{CREBBP} & -0.566 & 0.242 & 0.019 & (-1.040, -0.092) & -1.504 & 0.623 & 0.016 & (-2.726, -0.282) \\ 
				AX-11235551 & 16:3832471 & \textit{CREBBP} & -0.130 & 0.057 & 0.022 & (-0.242, -0.019) & -0.495 & 0.309 & 0.110 & (-1.101, 0.112) \\ 
				AX-11639833 & 5:88088439 & \textit{MEF2C} & -0.121 & 0.056 & 0.031 & (-0.231, -0.011) & -0.145 & 0.120 & 0.228 & (-0.381, 0.091) \\ 
				AX-11326149 & 15:78867482 & \textit{CHRNA5} & 0.102 & 0.051 & 0.046 & (0.002, 0.202) & 1.273 & 0.366 & 0.001 & (0.555, 1.991) \\ 
				AX-11376755 & 21:16340289 & \textit{NRIP1} & -0.101 & 0.052 & 0.052 & (-0.202, 0.001) & -0.281 & 0.120 & 0.019 & (-0.516, -0.046) \\ 
				AX-40181207 & 17:41218805 & \textit{BRCA1} & -0.524 & 0.272 & 0.054 & (-1.056, 0.009) & -2.386 & 0.750 & 0.001 & (-3.856, -0.916) \\ 
				AX-30854303 & 12:66761377 & \textit{GRIP1} & 0.094 & 0.054 & 0.081 & (-0.011, 0.199) & 0.102 & 0.117 & 0.380 & (-0.126, 0.331) \\ 
				$\vdots$ & & & & & & & & & & \\
				\hline
			\end{tabular}
			\begin{tablenotes}
				\item Est: coefficient estimate; SE: standard error estimate;  CI: confidence interval; HS: high school; AD: Adenocarcinoma; SCC: squamous cell carcinoma; LCC: large cell carcinoma; Pos: physical location based on Assembly GRCh37/hg19.
			\end{tablenotes}
		\end{threeparttable}
		\label{tab:est2}
	\end{table}
\end{landscape}

Among various patient characteristics, QP found that the adenocarcinoma subtype is significantly associated with better patient survival than large cell carcinoma, consistent with the results of \cite{janssen2001trends}, which was, however, not detected by  MPLE.  QP further identified that AX-11672686 in \textit{CHRNA2}, AX-11673610 in \textit{GRIP2} and AX-11264571 in \textit{BRCA2} are the three most significant SNPs associated with NSCLC patient survival, after adjusting for all the other demographic and genetic risk factors. Interestingly,  AX-11672686 was found to be associated with nicotine dependence by \cite{wang2014significant}. 
%\textit{GRIP1} has not been reported in the recent  literature review of a large amount of lung cancer studies by \citet{bosse2018decade}. %\citet{hershberger2005regulation} showed that ``NSCLC cells express proteins necessary to generate a transcriptional response to estrogen and suggest that ER$\beta$ and GRIP1 are likely mediators of this response". \bincm{What do you want to say here? GRIP1 is not reported but was reported?} \lucm{I wanted to say that the review paper \citet{bosse2018decade}, which reviewed quite a large number of survival studies on lung cancer, did not report the gene GRIP1. But, to justify the significance of one SNP we found in this gene, I found this paper \citet{hershberger2005regulation} saying the gene GRIP1 might play some mediator role from proteins to transcriptional response to estrogen. Since the mechanism is not very clear, we can delete the citation \citet{hershberger2005regulation}.} 
 AX-11264571 has been found to be associated with breast cancer \citep{qiu2010brca2} and may also be associated with lung cancer susceptibility, although not achieving genome-wide significance in \citet{yu2011analysis}. 
 AX-11673610 or \textit{GRIP1} seems to be a new finding as, to our knowledge, they  have yet been reported in the lung cancer literature  \citep{bosse2018decade}
%QP also identifies four other SNPs with non-adjusted 95\% confidence intervals excluding zero, two of which MPLE does not find. Two SNPs signaled by QP at level 0.05 are located in \textit{CREBBP}, which is  one of the most frequently mutated genes in small cell lung cancer \citep{jia2018crebbp}. 
%AX-11673610 in \textit{GRIP1} is the most significant SNP identified by both methods, and also the only significant one after Bonferroni correction. \textit{GRIP1} has not been reported in the literature review \citep{bosse2018decade}. \citet{hershberger2005regulation} showed that ``NSCLC cells express proteins necessary to generate a transcriptional response to estrogen and suggest that ER$\beta$ and GRIP1 are likely mediators of this response". Our proposed method \textsc{qp} also identified three other SNPs with non-adjusted 95\% CIs excluding zero, which \textsc{mple} did not. One of them https://www.overleaf.com/project/60a34aca15ae3f3dc8c9f5fcwas AX-41920413 in \textit{TBP}, and increased \textit{TBP} mRNA levels were previously found in some human breast and lung carcinoma tissues compared to normal tissues \citep{wada1992general}. 

To understand the impact of the socioeconomic status on cancer survival, we test for the association between   education level (no high school, high school, or  at least 1--2 years of college) and lung cancer patient survival. With a loading matrix  $A_{2\times p}= (e_2, e_3)^T$ corresponding to the contrast of the effects of high school graduate and at least 1--2 years of college with the reference level of no high school, the test statistic  is 0.259 with a p-value of 0.879, suggesting no statistical evidence for the association between education level and NSCLC patient survival, after adjusting all other  demographic characteristics and genetic  markers. {The results confirm a large-scale clinical trial on lung cancer patients which reported  ``education level was not predictive of survival" %\lu{regardless of adjustment for known prognostic factors}
\citep{herndon2008patient}.}

In summary, these results illustrate the utility of our method  in providing reliable inference for scientific discovery and interpretation, while more in-depth biological investigations are warranted to validate our findings.
%although the actual functions of genetic variants would need further biological investigations. 

%==============================================
\section{Concluding remarks}
\label{sec:conclusion}

%Motivated by the work of \citet{javanmard2014confidence},
{We have proposed a debiased lasso approach for reliable estimation and inference in the Cox proportional hazards model when $p<n$ but is allowed to diverge to
$\infty$ with $n$. Unlike existing methods  \citep{fang2017testing, yu2018confidence, kong2018high}, we resort a quadratic programming procedure to estimate the inverse information matrix, without imposing an unrealistic sparsity assumption on it.
%for the inverse information matrix $\Theta_{\beta^0}$.
The proposed debiased lasso estimator is asymptotically unbiased and normally distributed under  mild regularity conditions.
%Unlike existing methods in \citet{fang2017testing, yu2018confidence, kong2018high}, by exploiting the quadratic programming procedure, we do not require the unrealistic sparsity assumption for the inverse information matrix $\Theta_{\beta^0}$. 
%We have illustrated that the proposed debiased lasso estimator
{Our simulations demonstrate that, when $p$ is very small, the proposed method behaves similarly to the conventional MPLE; when $p$ is relatively large, it outperforms the competitors in  bias correction and  confidence interval coverage.}}

Lastly, we  touch upon the important issue of drawing inference with
$p>n$, though not a main focus of this paper. First, several methods \citep{fang2017testing,yu2018confidence,kong2018high} had been developed for handling ``$p>n$" inference problems; however, our analytical and simulation studies have pinpointed their possible limitations in providing sufficient bias correction and reliable confidence intervals even within the ``large $n$, diverging $p$" framework, likely due to the  sparsity assumptions on the inverse information matrix that may not hold in survival settings. One possible solution, by going beyond the de-biased lasso framework, is to perform repeated data splitting for model selection and estimation on two separate parts of the data and smooth the resulting estimates from multiple splits; see \citet{fei2021estimation} for inference on high dimensional generalized linear models. The validity of the method hinges upon the  sure screening property for the initial model selection, and we will explore its use in a survival setting in the future.

\section*{Acknowledgements}
We thank David C. Christiani, Qianyu Yuan and Mulong Du for sharing and discussing the BLCSC data. 
%Interpretation of the results is the sole responsibility of the authors. 
This work was supported in part by grants from the  National Institutes of Health 
(grant number: R01AG056764, R01CA249096, U01CA209414)
and the National Science Foundation (grant number: DMS 1915711).
%\fundinginfo{National Institutes of Health, Grant/Award Number: R01AG056764, R01CA249096, U01CA209414; National Science Foundation, Grant/Award Number: DMS 1915711}

%\section*{conflict of interest}

%\printendnotes

% Submissions are not required to reflect the precise reference formatting of the journal (use of italics, bold etc.), however it is important that all key elements of each reference are included.

\bibliographystyle{apalike}
\bibliography{sample}

\begin{thebibliography}{}

\bibitem[Andersen and Gill, 1982]{andersen1982cox}
Andersen, P.~K. and Gill, R.~D. (1982).
\newblock Cox's regression model for counting processes: {A} large sample
  study.
\newblock {\em The Annals of Statistics}, 10(4):1100--1120.

\bibitem[Antoniadis et~al., 2010]{antoniadis2010dantzig}
Antoniadis, A., Fryzlewicz, P., and Letu{\'e}, F. (2010).
\newblock The {Dantzig} selector in {Cox's} proportional hazards model.
\newblock {\em Scandinavian Journal of Statistics}, 37(4):531--552.

\bibitem[Boss{\'e} and Amos, 2018]{bosse2018decade}
Boss{\'e}, Y. and Amos, C.~I. (2018).
\newblock A decade of {GWAS} results in lung cancer.
\newblock {\em Cancer Epidemiology, Biomarkers \& Prevention}, 27(4):363--379.

\bibitem[Cai et~al., 2011]{cai2011constrained}
Cai, T., Liu, W., and Luo, X. (2011).
\newblock A constrained $\ell_1$ minimization approach to sparse precision
  matrix estimation.
\newblock {\em Journal of the American Statistical Association},
  106(494):594--607.

\bibitem[Cox, 1972]{cox1972regression}
Cox, D.~R. (1972).
\newblock Regression models and life-tables.
\newblock {\em Journal of the Royal Statistical Society: Series B
  (Methodological)}, 34(2):187--202.

\bibitem[Fan and Li, 2002]{fan2002variable}
Fan, J. and Li, R. (2002).
\newblock Variable selection for {Cox's} proportional hazards model and frailty
  model.
\newblock {\em The Annals of Statistics}, 30(1):74--99.

\bibitem[Fang et~al., 2017]{fang2017testing}
Fang, E.~X., Ning, Y., and Liu, H. (2017).
\newblock Testing and confidence intervals for high dimensional proportional
  hazards models.
\newblock {\em Journal of the Royal Statistical Society: Series B (Statistical
  Methodology)}, 79(5):1415--1437.

\bibitem[Fei and Li, 2021]{fei2021estimation}
Fei, Z. and Li, Y. (2021).
\newblock Estimation and inference for high dimensional generalized linear
  models: A splitting and smoothing approach.
\newblock {\em Journal of Machine Learning Research}, 22(58):1--32.

\bibitem[Gui and Li, 2005]{gui2005penalized}
Gui, J. and Li, H. (2005).
\newblock Penalized {Cox} regression analysis in the high-dimensional and
  low-sample size settings, with applications to microarray gene expression
  data.
\newblock {\em Bioinformatics}, 21(13):3001--3008.

\bibitem[Herndon et~al., 2008]{herndon2008patient}
Herndon, J.~E., II, A. B.~K., Holland, J.~C., and Paskett, E.~D. (2008).
\newblock Patient education level as a predictor of survival in lung cancer
  clinical trials.
\newblock {\em Journal of clinical oncology}, 26(25):4116.

\bibitem[Houston et~al., 2018]{houston2018histologic}
Houston, K.~A., Mitchell, K.~A., King, J., White, A., and Ryan, B.~M. (2018).
\newblock Histologic lung cancer incidence rates and trends vary by
  race/ethnicity and residential county.
\newblock {\em Journal of Thoracic Oncology}, 13(4):497--509.

\bibitem[Huang et~al., 2013]{huang2013oracle}
Huang, J., Sun, T., Ying, Z., Yu, Y., and Zhang, C.-H. (2013).
\newblock Oracle inequalities for the lasso in the {Cox} model.
\newblock {\em Annals of Statistics}, 41(3):1142--1165.

\bibitem[Janssen-Heijnen and Coebergh, 2001]{janssen2001trends}
Janssen-Heijnen, M.~L. and Coebergh, J.-W.~W. (2001).
\newblock Trends in incidence and prognosis of the histological subtypes of
  lung cancer in {North America, Australia, New Zealand and Europe}.
\newblock {\em Lung Cancer}, 31(2-3):123--137.

\bibitem[Javanmard and Montanari, 2014]{javanmard2014confidence}
Javanmard, A. and Montanari, A. (2014).
\newblock Confidence intervals and hypothesis testing for high-dimensional
  regression.
\newblock {\em Journal of Machine Learning Research}, 15(1):2869--2909.

\bibitem[Kong and Nan, 2014]{kong2014non}
Kong, S. and Nan, B. (2014).
\newblock Non-asymptotic oracle inequalities for the high-dimensional {Cox}
  regression via lasso.
\newblock {\em Statistica Sinica}, 24(1):25--42.

\bibitem[Kong et~al., 2018]{kong2018high}
Kong, S., Yu, Z., Zhang, X., and Cheng, G. (2018).
\newblock High dimensional robust inference for {Cox} regression models.
\newblock {\em arXiv preprint arXiv:1811.00535}.

\bibitem[McKay et~al., 2017]{mckay2017large}
McKay, J.~D., Hung, R.~J., Han, Y., Zong, X., Carreras-Torres, R., Christiani,
  D.~C., Caporaso, N.~E., Johansson, M., Xiao, X., Li, Y., et~al. (2017).
\newblock Large-scale association analysis identifies new lung cancer
  susceptibility loci and heterogeneity in genetic susceptibility across
  histological subtypes.
\newblock {\em Nature Genetics}, 49(7):1126--1132.

\bibitem[Ning and Liu, 2017]{ning2017general}
Ning, Y. and Liu, H. (2017).
\newblock A general theory of hypothesis tests and confidence regions for
  sparse high dimensional models.
\newblock {\em The Annals of Statistics}, 45(1):158--195.

\bibitem[Purcell et~al., 2007]{purcell2007plink}
Purcell, S., Neale, B., Todd-Brown, K., Thomas, L., Ferreira, M.~A., Bender,
  D., Maller, J., Sklar, P., De~Bakker, P.~I., Daly, M.~J., and Sham, P.
  (2007).
\newblock {PLINK}: a tool set for whole-genome association and population-based
  linkage analyses.
\newblock {\em The American Journal of Human Genetics}, 81(3):559--575.

\bibitem[Qiu et~al., 2010]{qiu2010brca2}
Qiu, L.-X., Yao, L., Xue, K., Zhang, J., Mao, C., Chen, B., Zhan, P., Yuan, H.,
  and Hu, X.-C. (2010).
\newblock {BRCA2} {N372H} polymorphism and breast cancer susceptibility: a
  meta-analysis involving 44,903 subjects.
\newblock {\em Breast Cancer Research and Treatment}, 123(2):487--490.

\bibitem[Tang et~al., 2020]{tang2020novel}
Tang, D., Zhao, Y.~C., Qian, D., Liu, H., Luo, S., Patz, E.~F., Moorman, P.~G.,
  Su, L., Shen, S., Christiani, D.~C., Glass, C., Gao, W., and Wei, Q. (2020).
\newblock Novel genetic variants in hdac2 and ppargc1a of the creb-binding
  protein pathway predict survival of non-small-cell lung cancer.
\newblock {\em Molecular Carcinogenesis}, 59(1):104--115.

\bibitem[Tibshirani, 1997]{tibshirani1997lasso}
Tibshirani, R. (1997).
\newblock The lasso method for variable selection in the {Cox} model.
\newblock {\em Statistics in Medicine}, 16(4):385--395.

\bibitem[van~de Geer et~al., 2014]{van2014asymptotically}
van~de Geer, S., B{\"u}hlmann, P., Ritov, Y., and Dezeure, R. (2014).
\newblock On asymptotically optimal confidence regions and tests for
  high-dimensional models.
\newblock {\em The Annals of Statistics}, 42(3):1166--1202.

\bibitem[van~der Vaart, 1998]{van1998asymptotic}
van~der Vaart, A.~W. (1998).
\newblock {\em Asymptotic {Statistics}}.
\newblock Cambridge: Cambridge University Press.

\bibitem[van~der Vaart and Wellner, 1996]{van1996weak}
van~der Vaart, A.~W. and Wellner, J.~A. (1996).
\newblock {\em {Weak Convergence and Empirical Processes: With Applications to
  Statistics}}.
\newblock Heidelberg: Springer.

\bibitem[Verweij and van Houwelingen, 1993]{verweij1993cross}
Verweij, P.~J. and van Houwelingen, H.~C. (1993).
\newblock Cross-validation in survival analysis.
\newblock {\em Statistics in mMdicine}, 12(24):2305--2314.

\bibitem[Wang et~al., 2014]{wang2014significant}
Wang, S., van~der Vaart, A.~D., Xu, Q., Seneviratne, C., Pomerleau, O.~F.,
  Pomerleau, C.~S., Payne, T.~J., Ma, J.~Z., and Li, M.~D. (2014).
\newblock Significant associations of {CHRNA2} and {CHRNA6} with nicotine
  dependence in {European} {American} and {African} {American} populations.
\newblock {\em Human Genetics}, 133(5):575--586.

\bibitem[Xia et~al., 2020]{xia2020revisit}
Xia, L., Nan, B., and Li, Y. (2020).
\newblock A revisit to de-biased lasso for generalized linear models.
\newblock {\em arXiv preprint arXiv:2006.12778}.

\bibitem[Yu et~al., 2011]{yu2011analysis}
Yu, H., Zhao, H., Wang, L.-E., Han, Y., Chen, W.~V., Amos, C.~I., Rafnar, T.,
  Sulem, P., Stefansson, K., Landi, M.~T., Caporaso, N., Albanes, D., Thun, M.,
  McKay, J.~D., Brennan, P., Wang, Y., Houlston, R.~S., Spitz, M.~R., and Wei,
  Q. (2011).
\newblock An analysis of single nucleotide polymorphisms of 125 {DNA} repair
  genes in the {Texas} genome-wide association study of lung cancer with a
  replication for the {XRCC4} {SNPs}.
\newblock {\em DNA Repair}, 10(4):398--407.

\bibitem[Yu et~al., 2018]{yu2018confidence}
Yu, Y., Bradic, J., and Samworth, R.~J. (2018).
\newblock Confidence intervals for high-dimensional {Cox} models.
\newblock {\em arXiv preprint arXiv:1803.01150}.

\bibitem[Zhang and Zhang, 2014]{zhang2014confidence}
Zhang, C.-H. and Zhang, S.~S. (2014).
\newblock Confidence intervals for low dimensional parameters in high
  dimensional linear models.
\newblock {\em Journal of the Royal Statistical Society: Series B (Statistical
  Methodology)}, 76(1):217--242.

\end{thebibliography}

%%%%%%%%%%%%%%%%%%%%%%%%%%%%%%%%%%%%%%%%%%%%%%%%%
%%%%%%%%%%%%%% Appendix %%%%%%%%%%%%%%%%%%%%%%%%%%%%%%%%%%%%%%%%%%%%%%%%%
\section*{Appendix}

\renewcommand{\theequation}{A\arabic{equation}}
\setcounter{equation}{0}

\setcounter{theorem}{0}
\renewcommand{\thetheorem}{A\arabic{theorem}}

We first present the useful lemmas for proving the main theorems, with detailed proofs deferred to the online supplementary materials. Some of these lemmas present important results in their own right. The proofs of the Theorem \ref{thm:main} and Theorem \ref{thm:simul} are presented following the lemmas.

Additional notation from counting processes and martingale theory is defined  for the proofs. Under the Cox model, define the counting process $N_i(t) = 1(Y_i \le t, \delta_i = 1)$ and its compensator $A_i(t; \beta) = \int_0^t 1 (Y_i \ge s) \exp(X_i^T \beta) d H_0(s)$, where $H_0(t) = \int_{0}^{t} h_0(s) ds$ is the cumulative baseline hazard function, $i = 1, \cdots, n$. Let $M_i(t; \beta) = N_i(t) - A_i(t; \beta)$, and $M_i(t; \beta^0)$ is a martingale with respect to the filtration $\mathcal{F}_i(t) = \sigma \{ N_i(s), 1(Y_i \ge s), X_i: s \in (0, t] \}$. It follows that %$\{X_i - 
$\heta_n(t; {\beta)}$, and in particular, $\heta_n(t; {\beta^0)}$,  is predictable with respect to the filtration $\mathcal{F}(t) = \sigma \{ N_i(s), 1(Y_i \ge s), X_i: s \in (0, t], i=1, \cdots, n \}$, an observation useful for
derivations. Notation-wise, we do not distinguish between the usual expectation and the outer expectation.
%and all conclusions still hold.

%%% lemma 1: moment approximation %%%

Lemma \ref{lemma:mom} below characterizes the difference between $\widehat{\eta}_n(t; \beta^0)$ and $\eta_0(t; \beta^0)$, which {facilitates the proof} of the asymptotic distribution for the leading term $\sqrt{n} c^T \Theta_{\beta^0} \dot{\ell}_n(\beta^0)$ as well as  {the establishment} of the convergence rate for $\hSigma - \Sigma_{\beta^0}$. 

\begin{lemma} \label{lemma:mom}
	Under Assumptions 1--3, we have
	\begin{align*}
	& \sup_{t \in [0, \tau]} | \hmu_0(t; \beta^0) - \mu_0(t; \beta^0) |  = \OP(\sqrt{\log(p) / n}), \\
	& \sup_{t \in [0, \tau]} \| \hmu_1(t; \beta^0) - \mu_1(t; \beta^0) \|_{\infty}  = \OP(\sqrt{\log(p)/n}), \\
	& \sup_{t \in [0, \tau]} \| \heta_n(t; \beta^0) - \eta_0(t; \beta^0)  \|_{\infty}  = \OP(\sqrt{\log(p)/n}).
	\end{align*}
\end{lemma}

%%%%%% lemma for leading term %%%%%

Lemma \ref{lemma:lead} establishes the asymptotic distribution for the leading term $- c^T \Thetabeta \dot{\ell}_n(\beta^0)$ in the decomposition of $c^T (\widehat{b} - \beta^0)$.
%up to the rescaling with the standard deviation. 

\begin{lemma} \label{lemma:lead}
	Assume $p^2 \log(p) / n \rightarrow 0$. Under Assumptions 1--5, for any $c \in \mathbb{R}^p$ such that $\| c \|_2 = 1$ and $\| c \|_1 \le a_*$ with some absolute constant $a_* < \infty$,
	\[
	\displaystyle \frac{ \sqrt{n} c^T \Thetabeta \dot{\ell}_n(\beta^0)}{\sqrt{c^T \Thetabeta c}} \overset{\mathcal{D}}{\rightarrow} N(0,1).
	\]
\end{lemma}

%%%%% lemma for lasso estimator

Lemma \ref{lemma:lasso} provides theoretical properties of the lasso estimator in the Cox model. This is a direct result from Theorem 1 in \citet{kong2014non}, and thus the proof is omitted.

\begin{lemma} \label{lemma:lasso}
	Under Assumptions 1--5, for the lasso estimator $\hbeta$, we have
	\[
	\|  \hbeta - \beta^0 \|_1 = \OP(s_0 \lambda_n), \quad \frac{1}{n} \sum_{i=1}^n |X_i^T (\hbeta - \beta^0)|^2 = \OP(s_0 \lambda_n^2),
	\]
	where $s_0 = |\{j: \beta_j^0 \ne 0, j=1, \cdots, p \}|$ is the true model size.
\end{lemma}

%%%%%%%%% lemma 4: feasible solution %%%%%%%%%%%

\begin{lemma} \label{lemma:const_sol}
	Under Assumptions 1--5, if $\lambda_n \asymp \sqrt{\log(p)/n}$, with probability going to 1, we have  $\| \Thetabeta \widehat{\Sigma} - I_p \|_{\infty} \le \gamma_n$, for $\gamma_n \asymp \| \Thetabeta \|_{1,1} s_0 \lambda_n$. 
\end{lemma}

Lemma \ref{lemma:const_sol} shows that, unlike in a linear regression model where the tuning parameter in the constraint  takes the order of $\sqrt{\log(p) / n}$, the Cox model requires a potentially larger $\gamma_n$ for the feasibility of $\Thetabeta$ depending on $\| \Theta_{\beta^0} \|_{1,1}$, because the information matrix involves the regression coefficients.

%%%%% lemma 5 %%%%%%%%

\begin{lemma} \label{lemma:diff_theta}
	Assume $\limsup_{n \rightarrow \infty} p \gamma_n \le 1 - \epsilon^{\prime}$ for some $\epsilon^{\prime} \in (0,1)$. Then, under the assumptions in Lemma \ref{lemma:const_sol}, $\|  \widehat{\Theta} - \Theta_{\beta^0} \|_{\infty} = \OP(\gamma_n \| \Theta_{\beta^0} \|_{1,1})$.
\end{lemma}

%%% lemma 6 %%%%%%%

\begin{lemma} \label{lemma:max_score}
	Under Assumptions 1--3 and 5, for each $t > 0$, 
	\[
	\PP \left( \|  \dot{\ell}_n(\beta^0) \|_{\infty} > t  \right) \le 2 p e^{-nt^2 / (8 K^2)}.
	\]
\end{lemma}

Now we complete the proofs of Theorem \ref{thm:main} and Theorem \ref{thm:simul}.

\begin{proof}[\textit{Proof of Theorem \ref{thm:main}.}]

 The first order Taylor expansion of $\dot{\ell}_{nj}(\hbeta)$, the $j$th component in $\dot{\ell}_{n}(\hbeta)$, at $\beta^0$, is
\begin{equation} \label{eq:taylor}
\dot{\ell}_{nj} (\hbeta) = \dot{\ell}_{nj} (\beta^0) +  [\ddot{\ell}_{n j} (\widetilde{\beta}^{(j)})]^T ( \hbeta - \beta^0),
\end{equation}  
where $\widetilde{\beta}^{(j)}$ lies between $\hbeta$ and $\beta^0$, and $\ddot{\ell}_{n j} (\beta)$ denotes the $j$th column in the Hessian matrix $\ddot{\ell}_{n} (\beta)$. Let the $p \times p$ matrix $B_n = (  \ddot{\ell}_{n 1} (\widetilde{\beta}^{(1)}), \ldots, \ddot{\ell}_{n p} (\widetilde{\beta}^{(p)}) )^T$. Suppose $c \in \mathbb{R}^p$ is a $p$-dimensional vector, and the parameter of interest is $c^T \beta^0$.  Plugging (\ref{eq:taylor}) in (\ref{eq:dslasso}), we have 
\begin{align} \label{eq:decomp}
c^T (\widehat{b} - \beta^0) & =  - c^T \Thetabeta \dot{\ell}_n(\beta^0) - c^T (\widehat{\Theta} - \Thetabeta) \dot{\ell}_n(\beta^0) \nonumber \\
& \quad - c^T (\widehat{\Theta} \hSigma - I_p) (\hbeta - \beta^0) + c^T \hTheta (\hSigma -  B_n) (\hbeta - \beta^0).
\end{align} 
The first term in (\ref{eq:decomp}) is the leading part and is asymptotically normal as shown in Lemma \ref{lemma:lead}, and the others will be proved to be asymptotically negligible. 
	
	%\noindent
	First, we show that $\sqrt{n} c^T (\widehat{\Theta} - \Thetabeta) \dot{\ell}_n(\beta^0) = \oP(1)$.
	By Lemma \ref{lemma:diff_theta} and Lemma \ref{lemma:max_score},
	\begin{align*}
	| \sqrt{n} c^T (\widehat{\Theta} - \Thetabeta) \dot{\ell}_n(\beta^0) | & \le \sqrt{n} \| c\|_1 \cdot \| \widehat{\Theta} - \Thetabeta \|_{\infty, \infty} \cdot \| \dot{\ell}_n(\beta^0) \|_{\infty} \\
	& \le \sqrt{n} a_{*} \OP( p \gamma_n \| \Thetabeta \|_{1,1}) \OP(\sqrt{\log(p)/n}) \\
	& = \OP( \| \Thetabeta \|_{1,1} p \gamma_n \sqrt{\log(p)}  ) \\
	& = \oP(1).
	\end{align*}

	%\noindent
	Second, we show that $\sqrt{n} c^T (\widehat{\Theta} \hSigma - I_p) (\hbeta - \beta^0) = \oP(1)$.
	By Lemma \ref{lemma:lasso},
	\begin{align*}
	| \sqrt{n} c^T (\widehat{\Theta} \widehat{\Sigma} - I_p) (\hbeta - \beta^0) | & \le \sqrt{n} \| c\|_1 \| (\widehat{\Theta} \widehat{\Sigma} - I_p) (\hbeta - \beta^0) \|_{\infty} \\
	& \le \sqrt{n} a_* \| \hTheta \widehat{\Sigma}  - I_p \|_{\infty} \| \hbeta - \beta^0 \|_1 \\
	& \le \sqrt{n} a_* \gamma_n \| \hbeta - \beta^0 \|_1 \\
	& = \OP(\sqrt{n} \gamma_n s_0 \lambda_n) \\
	& = \oP(1).
	\end{align*}

	%\noindent
	Next, we show that $\sqrt{n} c^T \hTheta (\hSigma -  B_n) (\hbeta - \beta^0) = \oP(1)$.
	Note that 
	\begin{equation} \label{eq:mat_decomp}
	\hSigma - B_n = (\hSigma - \Sigma_{\beta^0}) + (\Sigma_{\beta^0} - \ddot{\ell}_n(\beta^0)) + (\ddot{\ell}_n(\beta^0) - B_n).
	\end{equation} By the proof of Lemma \ref{lemma:const_sol}, we see that with $\lambda_n \asymp \sqrt{\log(p) / n}$, $\| \hSigma - \Sigma_{\beta^0} \|_{\infty} = \OP(s_0 \lambda_n)$. 
	We rewrite
	\begin{align}
	\Sigma_{\beta^0} - \ddot{\ell}_n(\beta^0) & = \displaystyle \E \int^{\tau}_0 \{ X_i - \eta_0(t; \beta^0) \}^{\otimes 2} e^{X_i^T \beta^0} 1(Y_i \ge t) h_0(t) dt \nonumber \\
	& \quad - \int_0^{\tau} \left\{ \hmu_2(t; \beta^0) - \frac{\hmu_1(t; \beta^0) \hmu_1^T(t; \beta^0)}{\hmu_0(t; \beta^0)} \right\}  h_0(t) dt \nonumber \\
	& \quad - \frac{1}{n} \sum_{i=1}^n \int_{0}^{\tau} \left\{  \frac{\hmu_2(t; \beta^0)}{\hmu_0(t; \beta^0)} - \left[  \frac{\hmu_1(t; \beta^0)}{\hmu_0(t; \beta^0)}  \right]^{\otimes 2}  \right\} dM_i(t)  \nonumber\\ 
	& = \int_0^{\tau}  \{  \mu_2(t; \beta^0) - \hmu_2(t; \beta^0)  \} h_0(t) dt \nonumber \\
	& \quad + \int_0^{\tau}  \left\{ \frac{\hmu_1(t; \beta^0) \hmu_1^T(t; \beta^0)}{\hmu_0(t; \beta^0)} - \frac{\mu_1(t; \beta^0) \mu_1^T(t; \beta^0)}{\mu_0(t; \beta^0)}   \right\} h_0(t) dt \nonumber \\
	& \quad - \frac{1}{n} \sum_{i=1}^n \int_{0}^{\tau} \left\{  \frac{\hmu_2(t; \beta^0)}{\hmu_0(t; \beta^0)} - \left[  \frac{\hmu_1(t; \beta^0)}{\hmu_0(t; \beta^0)}  \right]^{\otimes 2}  \right\} dM_i(t).
	\label{eq:mid}
	\end{align}
	Similar to the proof in Lemma \ref{lemma:mom}, we can show that $\sup_{t \in [0, \tau]} \| \hmu_2(t; \beta^0) - \mu_2(t; \beta^0) \|_{\infty} = \OP(\sqrt{\log(p)/n})$, and thus $ \| \int_0^{\tau}  \{  \mu_2(t; \beta^0) - \hmu_2(t; \beta^0)  \} h_0(t) dt \|_{\infty} \le \sup_{t \in [0, \tau]} \| \hmu_2(t; \beta^0) - \mu_2(t; \beta^0) \|_{\infty} \int_{0}^{\tau} h_0(t) dt = \OP(\sqrt{\log(p)/n})$. Since 
	\[
	\displaystyle \frac{\hmu_1 \hmu_1^T}{\hmu_0} - \frac{\mu_1 \mu_1^T}{\mu_0} = \frac{\hmu_1 \hmu_1^T}{\hmu_0 \mu_0}(\mu_0 - \hmu_0) + \frac{1}{\mu_0} [ (\hmu_1 - \mu_1) \hmu_1^T + \mu_1 (\hmu_1 - \mu_1)^T ]
	\]
	in  the second term of  (\ref{eq:mid}), by Assumption 1 and Lemma \ref{lemma:mom}, 
	\[
	\left\| \int_0^{\tau}  \left\{ \frac{\hmu_1(t; \beta^0) \hmu_1^T(t; \beta^0)}{\hmu_0(t; \beta^0)} - \frac{\mu_1(t; \beta^0) \mu_1^T(t; \beta^0)}{\mu_0(t; \beta^0)}   \right\} h_0(t) dt  \right\|_{\infty} = \OP(\sqrt{\log(p)/n}).
	\]
	$ n^{-1} \sum_{i=1}^n \int_{0}^{\tau} \left\{  {\mu_2(t; \beta^0)}/{\mu_0(t; \beta^0)} - \left[  {\mu_1(t; \beta^0)}/{\mu_0(t; \beta^0)}  \right]^{\otimes 2}  \right\} dM_i(t)$ is a sum of $n$ independent and identically distributed mean zero terms, and each term $ \left\| \int_{0}^{\tau} \left\{  {\mu_2(t; \beta^0)}/{\mu_0(t; \beta^0)} - \left[  {\mu_1(t; \beta^0)}/{\mu_0(t; \beta^0)}  \right]^{\otimes 2}  \right\} dM_i(t) \right\|_{\infty}$ is bounded by $2K^2(1 + e^{K_1} H_0(\tau))$ uniformly for all $i$ and $t\in [0, \tau]$. Similar to the proof of $\| A_n \|_{\infty} = \OP(\sqrt{\log(p)/n})$ in Lemma \ref{lemma:const_sol}, by Hoeffding's concentration inequality, 
	\[
	\left\|  \displaystyle \frac{1}{n} \sum_{i=1}^n \int_{0}^{\tau} \left\{  \frac{\mu_2(t; \beta^0)}{\mu_0(t; \beta^0)} - \left[  \frac{\mu_1(t; \beta^0)}{\mu_0(t; \beta^0)}  \right]^{\otimes 2}  \right\} dM_i(t) \right\|_{\infty} = \OP(\sqrt{\log(p)/n}). \] 
	It is easy to see that 
	\[
	\sup_{t \in [0, \tau]} \left\| \left\{  \frac{\hmu_2(t; \beta^0)}{\hmu_0(t; \beta^0)} - \left[  \frac{\hmu_1(t; \beta^0)}{\hmu_0(t; \beta^0)}  \right]^{\otimes 2}  \right\} - \left\{  \frac{\mu_2(t; \beta^0)}{\mu_0(t; \beta^0)} - \left[  \frac{\mu_1(t; \beta^0)}{\mu_0(t; \beta^0)}  \right]^{\otimes 2}  \right\} \right\|_{\infty} = \OP \left(\sqrt{\displaystyle \frac{\log(p)}{n}} \right).
	\]
	Then \begin{align*}
	& \left\| \frac{1}{n} \sum_{i=1}^n \int_{0}^{\tau} \left\{  \frac{\hmu_2(t; \beta^0)}{\hmu_0(t; \beta^0)} - \left[  \frac{\hmu_1(t; \beta^0)}{\hmu_0(t; \beta^0)}  \right]^{\otimes 2}  \right\} dM_i(t) \right. \\
	& \left. - \frac{1}{n} \sum_{i=1}^n \int_{0}^{\tau} \left\{  \frac{\mu_2(t; \beta^0)}{\mu_0(t; \beta^0)} - \left[  \frac{\mu_1(t; \beta^0)}{\mu_0(t; \beta^0)}  \right]^{\otimes 2}  \right\} dM_i(t) \right\|_{\infty} = \OP \left(\sqrt{\displaystyle \frac{\log(p)}{n}} \right),
	\end{align*}
	and thus for the third term in (\ref{eq:mid}), 
	\[
	\displaystyle \left\| \frac{1}{n} \sum_{i=1}^n \int_{0}^{\tau} \left\{  \frac{\hmu_2(t; \beta^0)}{\hmu_0(t; \beta^0)} - \left[  \frac{\hmu_1(t; \beta^0)}{\hmu_0(t; \beta^0)}  \right]^{\otimes 2}  \right\} dM_i(t) \right\|_{\infty} = \OP(\sqrt{\log(p)/n}).
	\]
	Therefore, by (\ref{eq:mid}), $\| \Sigma_{\beta^0} - \ddot{\ell}_n(\beta^0) \|_{\infty} = \OP(\sqrt{\log(p)/n})$.

	For the $(j,k)$th element in $\ddot{\ell}_n(\beta)$, denoted as $\ddot{\ell}_{njk}(\beta)$, by the mean value theorem, we have
	\[
	\ddot{\ell}_{njk}(\widetilde{\beta}^{(j)}) - \ddot{\ell}_{njk}(\beta^0) =  \displaystyle  (\widetilde{\beta}^{(j)} - \beta^0)^T \left. \frac{\partial \ddot{\ell}_{njk}(\beta)}{\partial \beta} \right|_{\beta = \overline{\beta}^{(jk)}} ,
	\]
	where $\overline{\beta}^{(jk)}$ lies in the segment between $\widetilde{\beta}^{(j)}$ and $\beta^0$. Under Assumptions 1--3, when $\| \beta - \beta^0 \|_1 \le \delta^{\prime}$  for $\delta^{\prime} > 0$  small enough,  $  \left\| {\partial \ddot{\ell}_{njk}(\beta)}/{\partial \beta} \right\|_{\infty}$ is bounded by some constant related to $\delta^{\prime}$ uniformly for all $(j,k)$. Since $s_0 \lambda_n = o(1)$, we have $\|  B_n - \ddot{\ell}_n(\beta^0) \|_{\infty} \le \OP(\| \hbeta - \beta^0 \|_1) = \OP(s_0 \lambda_n)$.
	
	Combining the three parts in (\ref{eq:mat_decomp}), we have that for $\lambda_n \asymp \sqrt{\log(p)/n}$, $\| \hSigma - B_n \|_{\infty} = \OP(s_0 \lambda_n)$. Then 
	\begin{align*}
	| \sqrt{n} c^T \hTheta (\hSigma -  B_n) (\hbeta - \beta^0) | & \le \sqrt{n} \| c \|_1 \| \hTheta \|_{\infty, \infty} \| \hSigma - B_n \|_{\infty} \| \hbeta - \beta^0 \|_1 \\
	& \le \OP(\sqrt{n} \| \Theta_{\beta^0} \|_{1,1} (s_0 \lambda_n)^2 ) \\
	& = \oP(1).
	\end{align*}
	%\textcolor{red}{maybe need less restrictive, $s_0 \lambda^2$}.

	We show that the variance estimator is consistent, i.e.  $ c^T ( \widehat{\Theta}  - \Thetabeta) c \rightarrow_P 0$ as $n \rightarrow \infty$.
	\begin{align*}
	| c^T ( \widehat{\Theta}  - \Thetabeta) c | & \le \| c \|_1^2 \| \widehat{\Theta}  - \Thetabeta \|_{\infty} \\
	& \le a_*^2 \OP(\gamma_n \| \Thetabeta \|_{1,1}) = \oP(1). 
	\end{align*}
	Finally, by the arguments above and Slutsky's theorem, it holds that $\sqrt{n} c^T (\widehat{b} - \beta^0) / (c^T \hTheta c)^{1/2} \overset{\mathcal{D}}{\rightarrow} N(0,1)$.
\end{proof}

\begin{proof}[\textit{Proof of Theorem \ref{thm:simul}.}]
We prove Theorem \ref{thm:simul} using the Cram\'{e}r-Wold device. For any $\omega \in \mR^{l}$, where the dimension $l$ is a fixed integer free of $n$ and $p$, let $c=A^T \omega$ in Theorem 1. Essentially, we only require $\| c \|_1 = \| A^T \omega \|_{1}$ is upper bounded, and it is not essential to force $\| c \|_2 = 1$. Since $\| A \|_{\infty,\infty} =\mathcal{O}(1)$ (by assumption) and $\| \omega \|_1 =\mathcal{O}(1)$ (fixed $l$), then $\| A^T \omega \|_1 \le \| A^T \|_{1,1} \| \omega \|_1 = \| A \|_{\infty,\infty} \| \omega \|_1 = \mathcal{O}(1).$ 
\end{proof}

%%%%%%%%%%%%%%%%%%%%%%%%%%%%%
%%%%%%%%%%%% supplementary %%
%%%%%%%%%%%%%%%%%%%%%%%%%%%%%

\newpage

\setcounter{equation}{0}
\renewcommand{\theequation}{S\arabic{equation}}
\setcounter{theorem}{0}
\renewcommand{\thetheorem}{A\arabic{theorem}}
\setcounter{section}{0}
\renewcommand{\thesection}{S\arabic{section}}

\begin{center}
\Large \bf Supplementary Materials for ``Statistical Inference for Cox Proportional Hazards Models with a Diverging Number of Covariates"
\end{center}

We provide detailed proofs for the lemmas presented in the Appendix of the article, as well as patient characteristics of the Boston Lung Cancer Study Cohort data analyzed in Section 5.

\section{Technical proofs for the lemmas}

%%% lemma 1: moment approximation %%%

Lemma \ref{supp:lemma:mom}  characterizes the difference between $\widehat{\eta}_n(t; \beta^0)$ and $\eta_0(t; \beta^0)$, which is needed to prove the asymptotic distribution for the leading term $\sqrt{n} c^T \Theta_{\beta^0} \dot{\ell}_n(\beta^0)$ as well as to establish the convergence rate for $\hSigma - \Sigma_{\beta^0}$. 

\begin{lemma} \label{supp:lemma:mom}
	Under Assumptions 1--3, we have
	\begin{align*}
	& \sup_{t \in [0, \tau]} | \hmu_0(t; \beta^0) - \mu_0(t; \beta^0) |  = \OP(\sqrt{\log(p) / n}), \\
	& \sup_{t \in [0, \tau]} \| \hmu_1(t; \beta^0) - \mu_1(t; \beta^0) \|_{\infty}  = \OP(\sqrt{\log(p)/n}), \\
	& \sup_{t \in [0, \tau]} \| \heta_n(t; \beta^0) - \eta_0(t; \beta^0)  \|_{\infty}  = \OP(\sqrt{\log(p)/n}).
	\end{align*}
\end{lemma}

\begin{proof}[\textbf{Proof of Lemma \ref{supp:lemma:mom}}]
	The first two statements in the conclusion are similar to those in \citet{kong2014non}, but with differing setups. Consider a class of functions of $y \ge 0$ and $x \in \mR^p$ indexed by $t$, $\mathcal{F}_0 = \{ 1(y \ge t) \exp(x^T \beta^0): t \in [0, \tau] \}$. For any $0 < \epsilon < 1$, consider the cumulative distribution function for $Y$ and take an positive integer $m < 2/\epsilon$ and a sequence of  points $0 = t_0 < t_1 < \cdots < t_{m-1} < t_m = \infty$ such that $\PP(t_i < Y \le t_{i+1}) < \epsilon, ~ i = 0, 1, \ldots, m-1$. For each $i = 1, \cdots, m$, define the bracket $[L_i, U_i]$, where $L_i(x, y) = 1(y \ge t_i) \exp(x^T \beta^0)$ and $U_i(x, y) = 1(y > t_{i-1}) \exp(x^T \beta^0)$. We have $L_i(x,y) \le 1(y \ge t) \exp(x^T \beta^0) \le U_i(x,y)$ for $t_{i-1} < t \le t_i$, and 
	\begin{align*}
	& [ \E \{ U_i(X, Y) - L_i(X,Y) \}^2 ]^{1/2} = [ \E \{ 1(t_{i-1} < Y < t_i) \exp(2 X^T \beta^0) \} ]^{1/2} \le e^{K_1} \sqrt{\epsilon}, \\
	&   \E | U_i(X, Y) - L_i(X,Y) | = \E \{ 1(t_{i-1} < Y < t_i) \exp( X^T \beta^0) \} \le e^{K_1} \epsilon.
	\end{align*}
	Then the bracketing numbers \citet{van1998asymptotic} satisfy
	\[
	N_{[]}(e^{K_1} \sqrt{\epsilon}, \mathcal{F}_0, L_2(\PP)) \le \frac{2}{\epsilon}, \quad  N_{[]}(e^{K_1} {\epsilon}, \mathcal{F}_0, L_1(\PP)) \le \frac{2}{\epsilon},
	\]
	or equivalently, 
	\[
	N_{[]}(\epsilon, \mathcal{F}_0, L_2(\PP)) \le \frac{2 e^{2K_1}}{\epsilon^2}, \quad  N_{[]}({\epsilon}, \mathcal{F}_0, L_1(\PP)) \le \frac{2 e^{K_1}}{\epsilon} < \infty.
	\]
	By the Glivenko-Cantelli Theorem and the Donsker Theorem \citep{van1998asymptotic}, the class of $\mathcal{F}_0$ is  $\PP$-Glivenko-Cantelli and $\PP$-Donsker. So $\sup_{t \in [0, \tau]} | \hmu_0(t; \beta^0) - \mu_0(t; \beta^0) | \overset{a.s.}{\rightarrow} 0$, and moreover, by Theorem 2.14.9 of \citet{van1996weak} with $V=2$,
	\[
	\PP \left( \sqrt{n} \sup_{t \in [0, \tau]} | \hmu_0(t; \beta^0) - \mu_0(t; \beta^0) | > s \right) \le D e^{-s^2},
	\]
	for every $s > 0$ and a constant $D > 0$ that only depends on $K_1$. Setting $s = \sqrt{2\log(p)}$ implies that 
	\[
	\sup_{t \in [0, \tau]} | \hmu_0(t; \beta^0) - \mu_0(t; \beta^0) |  = \OP(\sqrt{\log(p) / n}).
	\]
	
	For the second statement, we consider the classes of functions of $(x, y) = (x_1, \cdots, x_p, y)$ indexed by $t$,
	\[
	\mathcal{F}_1^k = \{  1(y \ge t) e^{x^T \beta^0} x_k: t \in [0, \tau] \}, ~ k = 1, \cdots, p.
	\]
	Since $| e^{x^T \beta^0} x_k | \le K e^{K_1}$, similarly we have
	\[
	N_{[]}(\epsilon, \mathcal{F}_1^k, L_2(\PP)) \le \left(  \frac{\sqrt{2}e^{K_1} K}{\epsilon}  \right)^2. 
	\]
	By Theorem 2.14.9 of \citet{van1996weak} with $V=2$, we have
	\[
	\PP \left( \sqrt{n} \sup_{t \in [0, \tau]} | \hmu_{1k}(t; \beta^0) - \mu_{1k}(t; \beta^0) | > s \right) \le D^{\prime} s^2 e^{-2s^2} \le D^{\prime} e^{-1} e^{-s^2}
	\]
	for every $s > 0$, where $D^{\prime}$ is a constant that only depends on $K$ and $K_1$, and $\hmu_{1k}$ and $\mu_{1k}$ are the $k$th components of $\hmu_1$ and $\mu_1$, respectively. Thus,
	\begin{align*}
	& \PP \left( \sqrt{n} \sup_{t \in [0, \tau]} \| \hmu_{1}(t; \beta^0) - \mu_{1}(t; \beta^0) \|_{\infty} > s \right) \\
	\le & ~ \PP \left(  \bigcup_{k=1}^p \left\{  \sqrt{n} \sup_{t \in [0, \tau]} | \hmu_{1k}(t; \beta^0) - \mu_{1k}(t; \beta^0) | > s \right\}  \right) \\
	\le & ~ p  D^{\prime} e^{-s^2}.
	\end{align*}
	For example, taking $s = \sqrt{2\log(p)}$ would complete the proof for $\sup_{t \in [0, \tau]} \| \hmu_1(t; \beta^0) - \mu_1(t; \beta^0) \|_{\infty}  = \OP(\sqrt{\log(p)/n})$.
	
	Finally, we rewrite
	\begin{align*}
	\heta_n(t; \beta^0) - \eta_0(t;\beta^0) & = \displaystyle \frac{\hmu_1(t; \beta^0)}{\hmu_0(t; \beta^0)} - \frac{\mu_1(t; \beta^0)}{\mu_0(t; \beta^0)} \\
	& = \displaystyle \frac{\hmu_1(t; \beta^0)}{\mu_0(t; \beta^0)} - \frac{\mu_1(t; \beta^0)}{\mu_0(t; \beta^0)} + \frac{\hmu_1(t; \beta^0)}{\mu_0(t; \beta^0)}  \left(   \displaystyle \frac{\mu_0(t; \beta^0)}{\hmu_0(t; \beta^0)} - 1 \right).
	\end{align*}
	By Assumptions 1--3, $\mu_0(t; \beta^0) \ge e^{-K_1} \pi_0 > 0$ and $\sup_{t \in [0, \tau]} \|  \hmu_1(t; \beta^0) \|_\infty = \OP(1)$. Also, since
	\[
	\inf_{t \in [0, \tau]} \hmu_0(t; \beta^0) \ge \mu_0(t; \beta^0) - | \hmu_0(t; \beta^0) - \mu_0(t; \beta^0) | \ge e^{-K_1} \pi_0 - \sup_{t \in [0, \tau]} | \hmu_0(t; \beta^0) - \mu_0(t; \beta^0) | > e^{-K_1} \frac{\pi_0}{2}
	\]
	 almost surely, we have
	\begin{align*}
	& \sup_{t \in [0, \tau]} \left\|  \frac{\hmu_1(t; \beta^0)}{\mu_0(t; \beta^0)}  \left(   \frac{\mu_0(t; \beta^0)}{\hmu_0(t; \beta^0)} - 1 \right) \right\|_\infty \\
	\le & ~ \sup_{t \in [0, \tau]} \left\|  \frac{\hmu_1(t; \beta^0)}{\mu_0(t; \beta^0)}  \right\|_\infty \cdot \sup_{t \in [0, \tau]} \left| \frac{\mu_0(t; \beta^0)}{\hmu_0(t; \beta^0)} - 1 \right| \\
	\le & ~ \OP(1) \sup_{t \in [0, \tau]} \left| \mu_0(t; \beta^0) - \hmu_0(t; \beta^0) \right| = \OP(\sqrt{\log(p) / n}).
	\end{align*}
	Therefore, 
	\begin{align*}
	\sup_{t \in [0, \tau]} \| \heta_n(t; \beta^0) - \eta_0(t; \beta^0)  \|_{\infty} & \le \sup_{t \in [0, \tau]} \left\|  \displaystyle \frac{1}{\mu_0(t; \beta^0)} \left( \hmu_1(t; \beta^0) - \mu_1(t; \beta^0) \right) \right\|_{\infty} \\
	& \quad + \sup_{t \in [0, \tau]} \left\|  \frac{\hmu_1(t; \beta^0)}{\mu_0(t; \beta^0)}  \left(   \frac{\mu_0(t; \beta^0)}{\hmu_0(t; \beta^0)} - 1 \right) \right\|_\infty \\
	& = \OP(\sqrt{\log(p)/n}). 		
	\end{align*}     
\end{proof}

%%%%%% lemma for leading term %%%%%

Lemma \ref{supp:lemma:lead} establishes the asymptotic distribution for the leading term $- c^T \Thetabeta \dot{\ell}_n(\beta^0)$ in the decomposition of $c^T (\widehat{b} - \beta^0)$.
%up to rescaling with the standard deviation. 

\begin{lemma} \label{supp:lemma:lead}
	Assume $p^2 \log(p) / n \rightarrow 0$. Under Assumptions 1--5, for any $c \in \mathbb{R}^p$ such that $\| c \|_2 = 1$ and $\| c \|_1 \le a_*$ with some absolute constant $a_* > 0$,
	\[
	\displaystyle \frac{ \sqrt{n} c^T \Thetabeta \dot{\ell}_n(\beta^0)}{\sqrt{c^T \Thetabeta c}} \overset{\mathcal{D}}{\rightarrow} N(0,1).
	\]
\end{lemma}

\begin{proof}[\textbf{Proof of Lemma \ref{supp:lemma:lead}}]
	Using notation of martingales, we rewrite
	\begin{align}
	\displaystyle \frac{- \sqrt{n} c^T \Thetabeta \dot{\ell}_n(\beta^0)}{\sqrt{c^T \Thetabeta c}}  & =  \displaystyle \frac{1}{\sqrt{n}} \sum_{i=1}^n \frac{c^T \Thetabeta}{\sqrt{c^T \Thetabeta c}}  \left\{  X_i - \frac{\hmu_1(Y_i; \beta^0)}{\hmu_0(Y_i; \beta^0)}  \right\} \delta_i \nonumber \\
	& =  \displaystyle \frac{1}{\sqrt{n}} \sum_{i=1}^n \int_{0}^{\tau} \frac{c^T \Thetabeta}{\sqrt{c^T \Thetabeta c}}  \left\{  X_i - \frac{\hmu_1(t; \beta^0)}{\hmu_0(t; \beta^0)}  \right\} dN_i(t) \nonumber \\
	& =  \displaystyle \frac{1}{\sqrt{n}} \sum_{i=1}^n \int_{0}^{\tau} \frac{c^T \Thetabeta}{\sqrt{c^T \Thetabeta c}}  \left\{  X_i - \frac{\hmu_1(t; \beta^0)}{\hmu_0(t; \beta^0)}  \right\} dM_i(t). \nonumber
	\end{align}
	Let $Q_i(t) = \displaystyle \frac{1}{\sqrt{n}} \frac{c^T \Theta_{\beta^0}}{\sqrt{c^T \Theta_{\beta^0} c}} \left\{ X_i - \frac{\hmu_1(t; \beta^0)}{\hmu_0(t; \beta^0)} \right\}, ~i=1, \ldots, n$, which are predictable with respect to the filtration $\mathcal{F}$. Then
	\begin{equation} \label{eq:mart}
	\displaystyle \frac{- \sqrt{n} c^T \Thetabeta \dot{\ell}_n(\beta^0)}{\sqrt{c^T \Thetabeta c}}  =  \sum_{i=1}^n \int_0^{\tau} Q_i(t) dM_i(t).
	\end{equation}
	For any $t \in [0, \tau]$, let $U(t) = \sum_{i=1}^n \int_0^{t} Q_i(u) dM_i(u)$, whose predictable variation process is 
	\begin{align*}
	\langle U \rangle (t) & = \sum_{i=1}^n \int_{0}^{t} Q_i(u)^2 1(Y_i \ge u) e^{X_i^T \beta^0} d H_0(u) \\
	& =  \sum_{i=1}^n \int_0^t \frac{c^T \Theta_{\beta^0}}{c^T \Theta_{\beta^0} c} \left\{ X_i - \frac{\hmu_1(t; \beta^0)}{\hmu_0(t; \beta^0)} \right\} ^{\otimes 2}  \Theta_{\beta^0} c 1(Y_i \ge u) e^{X_i^T \beta^0} d H_0(u) \\
	& = \frac{c^T \Theta_{\beta^0}}{c^T \Theta_{\beta^0} c}  \left[ \int_0^t  \left\{ \hmu_2(u; \beta^0) - \frac{\hmu_1(u; \beta^0) \hmu_1(u; \beta^0)^T}{\hmu_0(u; \beta^0)} \right\} d H_0(u) \right] \Theta_{\beta^0} c 
	\end{align*}
	Similar to the proof in Lemma \ref{supp:lemma:mom}, we can show that $\sup_{t \in [0, \tau]} \| \hmu_2(t; \beta^0) - \mu_2(t; \beta^0) \|_{\infty} = \OP(\sqrt{\log(p)/n})$, and thus
	\begin{align} \label{eq:var_diff1}
	\left\| \int_0^{t}  \{  \mu_2(u; \beta^0) - \hmu_2(u; \beta^0)  \} h_0(u) du \right\|_{\infty} & \le \sup_{u \in [0, \tau]} \| \hmu_2(u; \beta^0) - \mu_2(u; \beta^0) \|_{\infty} \int_{0}^{\tau} h_0(u) du \nonumber \\
	& = \OP(\sqrt{\log(p)/n}).
	\end{align}
	Since 
	\[
	\displaystyle \frac{\hmu_1 \hmu_1^T}{\hmu_0} - \frac{\mu_1 \mu_1^T}{\mu_0} = \frac{\hmu_1 \hmu_1^T}{\hmu_0 \mu_0}(\mu_0 - \hmu_0) + \frac{1}{\mu_0} [ (\hmu_1 - \mu_1) \hmu_1^T + \mu_1 (\hmu_1 - \mu_1)^T ],
	\]
	by Assumption 1 and Lemma \ref{supp:lemma:mom}, 
	\begin{equation} \label{eq:var_diff2}
	\left\| \int_0^{t}  \left\{ \frac{\hmu_1(u; \beta^0) \hmu_1^T(u; \beta^0)}{\hmu_0(u; \beta^0)} - \frac{\mu_1(u; \beta^0) \mu_1^T(u; \beta^0)}{\mu_0(u; \beta^0)}   \right\} h_0(u) du  \right\|_{\infty} = \OP(\sqrt{\log(p)/n}).
	\end{equation}
	Combining \eqref{eq:var_diff1} and \eqref{eq:var_diff2}, we have that, uniformly for all $t \in [0, \tau]$,
	\begin{align*}
	 \left\|  \int_0^t  \left\{ \hmu_2(u; \beta^0) - \frac{\hmu_1(u; \beta^0) \hmu_1(u; \beta^0)^T}{\hmu_0(u; \beta^0)} \right\} dH_0(u) - \right. & \\
	\quad \left. \int_0^t  \left\{ \mu_2(u; \beta^0) - \frac{\mu_1(u; \beta^0) \mu_1(u; \beta^0)^T}{\mu_0(u; \beta^0)} \right\} d H_0(u)  \right\|_{\infty} & = \OP(\sqrt{\log(p)/n}). 
	\end{align*}
	Then
	\begin{align*}
	& \left| \langle U \rangle (t)  - \frac{c^T \Theta_{\beta^0}}{c^T \Theta_{\beta^0} c}  \left[ \int_0^t  \left\{ \mu_2(u; \beta^0) - \frac{\mu_1(u; \beta^0) \mu_1(u; \beta^0)^T}{\mu_0(u; \beta^0)} \right\} d H_0(u) \right] \Theta_{\beta^0} c   \right| \\
	\le & \zeta_{\mathrm{min}}^{-1} (\| c \|_1 \| \Theta_{\beta^0} \|_{1,1})^2 \OP(\sqrt{\log(p)/n}) \\
	\le & \zeta_{\mathrm{min}}^{-1} a_*^2 p \zeta_{\mathrm{max}}^2  \OP(\sqrt{\log(p)/n}) \rightarrow_P 0
	\end{align*}
	if $p^2 \log(p) / n \rightarrow 0$.  
	By Assumption 4, $\langle U(t) \rangle \rightarrow_P v(t; c)$. 
	
	Now we check the Lindeberg condition. For any $\epsilon > 0$, define the truncated process 
	\[
	U_{\epsilon}(t) = \sum_{i=1}^n \int_0^t Q_i(u) 1\{ |Q_i(u)| > \epsilon \} dM_i(u),
	\]
	with a predictable variation process:
	\begin{align*}
	\langle U_{\epsilon} \rangle(t) & = \sum_{i=1}^n \int_0^t Q_i^2(u) 1\{ |Q_i(u)| > \epsilon \} 1(Y_i \ge u) e^{X_i^T \beta^0} h_0(u) du \\
	& = \sum_{i=1}^n \int_0^t Q_i^2(u) 1\{ |\sqrt{n} Q_i(u)| > \sqrt{n} \epsilon \} 1(Y_i \ge u) e^{X_i^T \beta^0} h_0(u) du.
	\end{align*}
	Let $Q_{\mathrm{max}} = \sup_{t \in [0, \tau]} \max_{1\le i \le n} |\sqrt{n} Q_i(t)|$, then $1\{ |\sqrt{n} Q_i(u)| > \sqrt{n} \epsilon \} \le 1\{ Q_{\mathrm{max}} > \sqrt{n} \epsilon \}$. By Assumption 1, 
	\[
	\sup_{t \in [0, \tau]} \max_{1\le i \le n} \left| \displaystyle \frac{c^T \Theta_{\beta^0}}{\sqrt{c^T \Theta_{\beta^0} c}} \left\{ X_i - \frac{\hmu_1(t; \beta^0)}{\hmu_0(t; \beta^0)} \right\} \right| \le \zeta_{\mathrm{min}}^{-1/2} \| c \|_1 \| \Theta_{\beta^0} \|_{1,1} 2K = \mathcal{O}(\sqrt{p}),
	\]
	and $Q_{\mathrm{max}} = \mathcal{O}(\sqrt{p})$. When $p/n \rightarrow 0$, $1\{ Q_{\mathrm{max}} > \sqrt{n} \epsilon \} = 0$ almost surely. Thus $\langle U_{\epsilon} \rangle (t) \rightarrow_P 0$. Finally, by the martingale central limit theorem, the asymptotic normality follows.
\end{proof}

%%%%% lemma for lasso estimator

Lemma \ref{supp:lemma:lasso} provides the theoretical properties of the lasso estimator in the Cox model. This is a direct result from Theorem 1 in \citet{kong2014non}, and thus the proof is omitted.

\begin{lemma} \label{supp:lemma:lasso}
	Under Assumptions 1--5, for the lasso estimator $\hbeta$, we have
	\[
	\|  \hbeta - \beta^0 \|_1 = \OP(s_0 \lambda_n), \quad \frac{1}{n} \sum_{i=1}^n |X_i^T (\hbeta - \beta^0)|^2 = \OP(s_0 \lambda_n^2),
	\]
	where $s_0 = |\{j: \beta_j^0 \ne 0, j=1, \cdots, p \}|$ is the true model size.
\end{lemma}

%%%%%%%%% lemma 4: feasible solution %%%%%%%%%%%

\begin{lemma} \label{supp:lemma:const_sol}
	Under Assumptions 1--5, if $\lambda_n \asymp \sqrt{\log(p)/n}$, with probability going to 1, we have  $\| \Thetabeta \widehat{\Sigma} - I_p \|_{\infty} \le \gamma_n$, for $\gamma_n \asymp \| \Thetabeta \|_{1,1} s_0 \lambda_n$. 
\end{lemma}

Lemma \ref{supp:lemma:const_sol} shows that, unlike  linear  models with the tuning parameter in the constraint  taking the order of $\sqrt{\log(p) / n}$, the Cox model requires a potentially larger $\gamma_n$ for the feasibility of $\Thetabeta$ that depends on $\| \Theta_{\beta^0} \|_{1,1}$, as the information matrix involves the regression coefficients.

\begin{proof}[\textbf{Proof of Lemma \ref{supp:lemma:const_sol}}]
	Write $A_n = \displaystyle \frac{1}{n} \sum_{i=1}^n \int_0^{\tau}  \left\{  X_i - \eta_0(t; \beta^0)  \right\}^{\otimes 2} dN_i(t) - \Sigma_{\beta^0}$.
	\begin{align*}
	\| \widehat{\Sigma} - \Sigma_{\beta^0} \|_{\infty} & \le \left\| \displaystyle \frac{1}{n} \sum_{i=1}^n \int_{0}^{\tau} \left[  \left\{ X_i - \heta_n(t; \hbeta) \right\}^{\otimes 2} - \left\{  X_i - \eta_0(t; \beta^0)  \right\}^{\otimes 2}   \right] dN_i(t)   \right\|_{\infty}  \\
	& \quad + \left\|  \displaystyle \frac{1}{n} \sum_{i=1}^n \int_0^{\tau}  \left\{  X_i - \eta_0(t; \beta^0)  \right\}^{\otimes 2} dN_i(t) - \Sigma_{\beta^0}  \right\|_{\infty} \\
	& \le \left\| \displaystyle \frac{1}{n} \sum_{i=1}^n \int_{0}^{\tau}  \left\{ X_i - \heta_n(t; \hbeta) \right\}  \left\{ \heta_n(t; \hbeta) - \eta_0(t; \beta^0) \right\}^T  dN_i(t)   \right\|_{\infty}  \\ 
	& \quad +  \left\| \displaystyle \frac{1}{n} \sum_{i=1}^n \int_{0}^{\tau}  \left\{  \heta_n(t; \hbeta) - \eta_0(t; \beta^0) \right\}  \left\{ X_i - \eta_0(t; \beta^0)  \right\}^T  dN_i(t)   \right\|_{\infty}    + \left\|  A_n \right\|_{\infty}.
	\end{align*}
	Note that for all $t \in [0, \tau]$, $\| X_i - \heta_n(t; \hat{\beta}) \|_{\infty} \le 2K$ and $\| X_i - {\eta}_0(t; \beta^0) \|_{\infty} \le 2K$. Then 
	\begin{align}
	\| \widehat{\Sigma} - \Sigma_{\beta^0} \|_{\infty} & \le \displaystyle \frac{4K}{n} \sum_{i=1}^n \int_0^{\tau} \| \heta_n(t; \hat{\beta}) - \eta_0(t; \beta^0) \|_{\infty} dN_i(t) + \| A_n \|_{\infty} \nonumber \\
	& \le \displaystyle \frac{4K}{n} \sum_{i=1}^n \int_0^{\tau} \| \heta_n(t; \hat{\beta}) - \heta_n(t; \beta^0) \|_{\infty} dN_i(t) \nonumber \\
	& \quad + \displaystyle \frac{4K}{n}  \sum_{i=1}^n \int_0^{\tau} \| \heta_n(t; {\beta^0}) - \eta_0(t; \beta^0) \|_{\infty} dN_i(t) + \| A_n \|_{\infty}. 
	\label{eq:diff_sigma}
	\end{align}
	By the mean value theorem, for the $j$th component in $\heta_n$ (denoted by $\heta_{nj}$), there exists some $\bar{\beta}^{(j)}$ lying inside the segments connecting $\hbeta$ and $\beta^0$ such that
	\[
	\heta_{nj}(t; \hbeta) = \heta_{nj}(t; \beta^0) + \left[  \displaystyle \left. \frac{\partial \heta_{nj}(t; \beta)}{\partial \beta} \right|_{\beta = \bar{\beta}^{(j)}}  \right]^T (\hbeta - \beta^0).
	\]
	Consider $\beta$ in a neighborhood of $\beta^0$, i.e. when $\| \beta - \beta^0 \|_1 \le \delta^{\prime}$ for some $\delta^{\prime} > 0$, $e^{X_i^T \beta} \le e^{|X_i^T \beta|} \le e^{|X_i^T \beta^0| + K \delta^{\prime}} \le e^{K_1 + K \delta^{\prime}}$, and $e^{X_i^T \beta} \ge e^{-|X_i^T \beta|} \ge e^{-K_1 - K \delta^{\prime}}$. Since $\{ 1(Y \ge t): t \in [0, \tau]  \}$ is $\PP$-Glivenko-Cantelli, $\sup_{t \in [0, \tau]} | \frac{1}{n} \sum_{i=1}^n 1(Y \ge t) - \PP(Y \ge t) | \overset{a.s.}{\rightarrow} 0$, and then uniformly for $t \in [0, \tau]$ and $\beta \in \{ \beta: \|\beta - \beta^0 \|_1 \le \delta^{\prime} \}$,
	\[
	\hmu_0(t; \beta) \ge \frac{1}{n} \sum_{i=1}^n 1(Y_i \ge t) e^{-K_1 - K \delta^{\prime}} \overset{a.s.}{\rightarrow} \PP(Y \ge t) e^{-K_1 - K \delta^{\prime}} \ge \frac{\pi_0}{2} e^{-K_1 - K \delta^{\prime}}.
	\]
	In this case, uniformly for $t\in [0, \tau]$ and $\beta \in \{ \beta: \|\beta - \beta^0 \|_1 \le \delta^{\prime} \}$,
	\begin{align*}
	\left\| \displaystyle \frac{\partial \heta_n(t; \beta)}{\partial \beta^T} \right\|_{\infty} & = \left\| \displaystyle \frac{\hmu_2(t; \beta) \hmu_0(t; \beta) - \hmu_1(t; \beta) \hmu_1(t; \beta)^T}{\hmu_0^2(t; \beta)} \right\|_{\infty} \\
	& ~ \le_{a.s.} \left( \frac{\pi_0}{2} e^{-K_1-K\delta^{\prime}} \right)^{-2} \left\{ e^{K_1 + K\delta^{\prime}}K^2 \cdot e^{K_1 + K\delta^{\prime}} + e^{2(K_1+K\delta^{\prime})} K^2 \right\} \\
	& ~ = \frac{8}{\pi_0^2} e^{4(K_1+K\delta^{\prime})} K^2 < \infty,
	\end{align*}
	i.e. $\left\| \displaystyle \frac{\partial \heta_n(t; \beta)}{\partial \beta^T} \right\|_{\infty} $ is uniformly bounded almost surely. When $s_0 \lambda_n \rightarrow 0$, we have $\| \heta_{n}(t; \hbeta) - \heta_{n}(t; \beta^0) \|_{\infty} \le \OP(\| \hbeta - \beta^0 \|_1) = \OP(s_0 \lambda_n)$ and the first term in (\ref{eq:diff_sigma}) is $\frac{4K}{n} \sum_{i=1}^n \int_0^{\tau} \| \heta_n(t; \hbeta) - \heta_n(t; \beta^0) \|_{\infty} dN_i(t) = \OP(s_0 \lambda_n)$. 
	
	For the second term in (\ref{eq:diff_sigma}), we use an argument from Lemma \ref{supp:lemma:mom} that $\sup_{t \in [0, \tau]} \| \heta_n(t; {\beta^0}) - \eta_0(t; \beta^0) \|_{\infty} =  \OP(\sqrt{\log(p)/n})$ and then have
	\begin{align*}
	& \displaystyle \frac{4K}{n}  \sum_{i=1}^n \int_0^{\tau} \| \heta_n(t; {\beta^0}) - \eta_0(t; \beta^0) \|_{\infty} dN_i(t) \\
	\le & \frac{4K}{n}  \sum_{i=1}^n \int_0^{\tau}  \sup_{t \in [0, \tau]} \| \heta_n(t; {\beta^0}) - \eta_0(t; \beta^0) \|_{\infty} dN_i(t) \\
	= & \OP(\sqrt{\log(p)/n}).
	\end{align*}
	For the last term $A_n$, by Hoeffding's concentration inequality, we have for every $t > 0$ and $j, k = 1, \cdots, p$,
	\[
	\PP \left(  |A_n(j,k)| \ge t  \right) \le 2 \exp\{ - n t^2 / C^{\prime} \},
	\]
	where $C^{\prime}$ is a constant only depending on $K^4$. Since $A_n$ is a symmetric matrix, 
	\begin{align*}
	\PP \left(  \| A_n \|_{\infty} \ge t  \right) & = \PP \left( \bigcup_{1 \le j \le p, j \le k \le p} | A_n(j,k) | \ge t  \right) \\
	& \le \sum_{j=1}^p \sum_{k=j}^p \PP \left(  |A_n(j,k)| \ge t  \right)  \\
	& \le p(p+1) \exp\{ - n t^2 / C^{\prime} \}.
	\end{align*}
	So $\| A_n \|_{\infty} = \OP \left( \sqrt{\log(p) /n } \right)$. Combining the three terms in (\ref{eq:diff_sigma}), we have $\| \widehat{\Sigma} - \Sigma_{\beta^0} \|_{\infty} \le \OP(s_0 \lambda_n + \sqrt{\log(p) / n})$.
	Finally, we conclude that
	\begin{align*}
	\| \Thetabeta \widehat{\Sigma} - I_p \|_{\infty} & \le \| \Thetabeta \|_{1,1} \| \widehat{\Sigma} - \Sigma_{\beta^0} \|_{\infty} \\
	& = \OP \left(  \| \Thetabeta \|_{1,1} s_0 \lambda_n +  \| \Thetabeta \|_{1,1} \sqrt{\log(p)/n}      \right).
	\end{align*}  
	
\end{proof}

%%%%% lemma 5 %%%%%%%%

\begin{lemma} \label{supp:lemma:diff_theta}
	Assume $\limsup_{n \rightarrow \infty} p \gamma_n \le 1 - \epsilon^{\prime}$ for some $\epsilon^{\prime} \in (0,1)$. Then, under the assumptions in Lemma \ref{supp:lemma:const_sol}, $\|  \widehat{\Theta} - \Theta_{\beta^0} \|_{\infty} = \OP(\gamma_n \| \Theta_{\beta^0} \|_{1,1})$.
\end{lemma}

\begin{proof}[\textbf{Proof of Lemma \ref{supp:lemma:diff_theta}}]
	Note that $\widehat{\Theta} - \Thetabeta = \widehat{\Theta} (I_p - \widehat{\Sigma} \Thetabeta) + (\widehat{\Theta} \widehat{\Sigma} - I_p) \Thetabeta$, then on the event $\{ \| \widehat{\Sigma} \Thetabeta - I_p \|_{\infty} \le \gamma_n \}$, we have
	\begin{align*}
	\| \widehat{\Theta} - \Thetabeta \|_{\infty} & \le \| \widehat{\Theta} \|_{\infty, \infty} \| I_p - \widehat{\Sigma} \Thetabeta \|_{\infty} + \| \widehat{\Theta}  \widehat{\Sigma} - I_p \|_{\infty} \| \Thetabeta \|_{1,1} \nonumber  \\
	& \le \gamma_n \| \widehat{\Theta} \|_{\infty, \infty} + \gamma_n \| \Thetabeta \|_{1,1}. 
	\end{align*}
	Since $\| \widehat{\Theta} \|_{\infty, \infty}  \le \| \widehat{\Theta} - \Thetabeta \|_{\infty, \infty} + \| \Thetabeta \|_{\infty, \infty} \le p \| \widehat{\Theta} - \Thetabeta \|_{\infty} + \| \Thetabeta \|_{1,1}$,
	we can obtain 
	\[
	\| \widehat{\Theta} - \Thetabeta \|_{\infty} \le  \gamma_n \left( p  \| \widehat{\Theta} - \Thetabeta \|_{\infty}  + \|\Thetabeta \|_{1,1} \right) + \gamma_n \|\Thetabeta \|_{1,1}.
	\]
	When $\limsup_{n \rightarrow \infty} \gamma_n p \le 1 - \epsilon^{\prime} < 1$, then for $n$ large enough,
	\[
	\| \widehat{\Theta} - \Thetabeta \|_{\infty} \le 2\gamma_n \| \Thetabeta \|_{1,1} / (1 - \gamma_n p) \asymp \gamma_n \| \Thetabeta \|_{1,1}.
	\]
	Therefore, by Lemma \ref{supp:lemma:const_sol}, $\| \widehat{\Theta} - \Thetabeta \|_{\infty} = \OP(\gamma_n \| \Thetabeta \|_{1,1})$. 
\end{proof}

%%% lemma 6 %%%%%%%

\begin{lemma} \label{supp:lemma:max_score}
	Under Assumptions 1--3 and 5, for each $t > 0$, 
	\[
	\PP \left( \|  \dot{\ell}_n(\beta^0) \|_{\infty} > t  \right) \le 2 p e^{-nt^2 / (8 K^2)}.
	\]
\end{lemma}

\begin{proof}[\textbf{Proof of Lemma \ref{supp:lemma:max_score}}]
	Noting that $\| X_i - \heta_n(t; \beta^0) \|_{\infty} \le 2K$ uniformly for all $i$, then Lemma \ref{supp:lemma:max_score} is a direct result of Lemma 3.3(ii) in \citet{huang2013oracle}. 
\end{proof}

\section{Boston Lung Cancer Study Cohort data}

Table \ref{tab:surv_pop} shows the patient characteristics for the subset of the Boston Lung Cancer Study Cohort data analyzed in Section 5. 

\begin{table}[ht!]
	\centering
	\caption{Characteristics of $n=561$ patients in the Boston Lung Cancer Study for survival analysis}
	\small
	\begin{threeparttable}
		\begin{tabular}{lll}
			\hline
			Variable & Category / Unit  & Count (\%) / Mean (SD)  \\ \hline
			Age & Years old &  60.0 (10.9) \\
			Race & Caucasian & 528 (94.1\%)  \\ 
			& Others & 33 (5.9\%)   \\
			Education & No high school & 79 (14.1\%)  \\
			& High school & 141 (25.1\%)  \\
			& At least 1-2 years of college & 341 (60.8\%)  \\		
			Gender & Male & 215 (38.3\%)   \\
			& Female & 346 (61.7\%)  \\
			Smoker & Current or recently quit & 508 (90.6\%)  \\
			& Never & 53 (9.4\%)  \\
			Histology & Adenocarcinoma & 360 (64.2\%)  \\
			& Squamous cell carcinoma & 115 (20.5\%)  \\
			& Large cell carcinoma & 45 (8.0\%)  \\
			& Unspecified & 41 (7.3\%) \\
			Stage\tnote{a} & Early & 243 (43.3\%)  \\
			& Late & 318 (56.7\%)  \\ 
			Surgery & No & 177 (31.6\%)   \\
			& Yes & 361 (64.3\%)  \\
			Chemotherapy & No & 300 (53.5\%)  \\
			& Yes & 238 (42.4\%)  \\
			Radiation & No & 332 (59.2\%)   \\
			& Yes & 206 (36.7\%)  \\
			Treatment record & Missing\tnote{b} & 23 (4.1\%)  \\ \hline
		\end{tabular}
		\begin{tablenotes}
			\footnotesize
			\item [a] Stages I and II classified as early stage, and stages III and IV as late stage.
			\item [b] No treatment information on surgery, chemotherapy or radiation available for these patients.
		\end{tablenotes}
	\end{threeparttable}
	\label{tab:surv_pop}
\end{table}

\end{document}